\documentclass[envcountsame,envcountsect,oribibl,12pt]{llncs}

\usepackage{ltl_mc}

\bigpages

\extendedtrue

\begin{document}
\title{%
The Tractability of Model-Checking for LTL:\\
    The Good, the Bad, and the Ugly Fragments\thanks{%
   Supported in part by DFG VO 630/6-1 and the Postdoc Programme of the German Academic Exchange Service (DAAD).}
}

\author{%
  Michael Bauland\inst{1}  \and
  Martin Mundhenk\inst{2}  \and
  Thomas Schneider\inst{3} \and \\
  Henning Schnoor\inst{4}  \and
  Ilka Schnoor\inst{5}     \and
  Heribert Vollmer\inst{5}%
}
\institute{%
  Knipp GmbH, Martin-Schmei\ss er-Weg 9, 44227 Dortmund, Germany                            \\
  \protect\url|Michael.BaulandATknipp.de|                                                            \and
  Informatik, Friedrich-Schiller-Universit\"{a}t, 07737 Jena, Germany                       \\
  \protect\url|mundhenkATcs.uni-jena.de|                                                             \and
  Computer Science, University of Manchester, Oxford Road, Manchester M13 9PL, UK           \\
  \protect\url|schneiderATcs.man.ac.uk|                                                              \and
  Theoret.\ Informatik, Christian-Albrechts-Universit\"at, 24098 Kiel, Germany              \\
  \protect\url|schnoorATti.informatik.uni-kiel.de|                                                   \and
  Theoret.\ Informatik, Universit\"at L\"{u}beck, 23538 L\"{u}beck, Germany                 \\
  \protect\url|schnoorATtcs.uni-luebeck.de|                                                         \and 
  Theoret.\ Informatik, Universit\"{a}t Hannover, Appelstr. 4, 30167 Hannover, Germany      \\
  \protect\url|vollmerATthi.uni-hannover.de|
}

\maketitle

\begin{abstract}
  In a seminal paper from 1985, Sistla and Clarke showed
  that the model-checking problem for Linear Temporal Logic (LTL) is either \NP-complete
  or \PSPACE-complete, depending on the set of temporal operators used.
  If, in contrast, the set of propositional operators is restricted, the complexity may decrease.
  This paper systematically studies the model-checking problem for LTL formulae over restricted sets
  of propositional and temporal operators. For almost all combinations of temporal and propositional
  operators, we determine whether the model-checking problem is tractable (in \PTIME) or
  intractable (\NP-hard). We then focus on the tractable cases, showing that they all are \NL-complete
  or even logspace solvable.
  This leads to a surprising gap in complexity between tractable and intractable cases.
  It is worth noting that our analysis covers an infinite set of problems, since
  there are infinitely many sets of propositional operators.
\end{abstract}


\section{Introduction}

  Linear Temporal Logic (LTL) has been proposed by Pnueli~\cite{pnu77}
  as a formalism to specify properties of parallel programs and
  concurrent systems, as well as to reason about their behaviour.
  Since then, it has been widely used for these purposes.
  Recent developments require reasoning
  tasks---such as deciding satisfiability, validity, or
  model checking---to be performed automatically.
  Therefore, decidability and computational complexity
  of the corresponding decision problems are of great interest.

  The earliest and fundamental source of complexity results for the satisfiability problem (SAT) and
  the model-checking problem (MC) of LTL is certainly Sistla and Clarke's paper
  \cite{sicl85}. They have established \PSPACE-completeness of SAT and MC for LTL with the
  temporal operators \F\ (eventually), \G\ (invariantly), \X\ (next-time), \U\ (until),
  and \S\ (since). They have also shown that these problems are \NP-complete for certain
  restrictions of the set of temporal operators. 
  This work was continued by Markey~\cite{mar04}.
  The results of Sistla, Clarke, and Markey imply
  that SAT and MC for LTL and a multitude of its fragments are intractable.
  In fact, they do not exhibit any tractable fragment.

  The fragments they consider are obtained by restricting the
  set of temporal operators and the use of negations.
  What they do not consider are arbitrary fragments of 
  temporal \emph{and} Boolean operators.
  For propositional logic, a complete analysis has been achieved by Lewis
  \cite{lew79}.
  He divides all infinitely many sets of Boolean operators into those with tractable (polynomial-time
  solvable) and intractable (\NP-complete) SAT problems.
  A similar systematic classification has been obtained by Bauland et al.\ in \cite{bss+07} for LTL.
  They divide fragments of LTL---determined by arbitrary combinations
  of temporal and Boolean operators---into those with polynomial-time
  solvable, \NP-complete, and \PSPACE-complete SAT problems.

  This paper continues the work on the MC problem for LTL.
  Similarly as in \cite{bss+07},
  the considered fragments are arbitrary combinations of temporal and Boolean operators.
  We will separate the MC problem for almost all LTL fragments
  into tractable (\ie, polynomial-time solvable) and intractable (\ie, \NP-hard) cases.
  This extends the work of Sistla and Clarke, and Markey \cite{sicl85,mar04},
  but in contrast to their results,
  we will exhibit many tractable fragments and exactly determine their computational complexity.
  Surprisingly, we will see that tractable cases for model checking are even very easy---that is,
  \NL-complete or even \LS-solvable.
  There is only one set of Boolean operators, consisting of the binary \XOR-operator,
  that we will have to leave open.
  This constellation has already proved difficult to handle in \cite{bss+07,bhss06},
  the latter being a paper where SAT for basic modal logics has been classified in a similar way.

  While the borderline between tractable and intractable fragments
  in \cite{lew79,bss+07} is quite easily recognisable (SAT for fragments containing the
  Boolean function $f(x,y) = x \wedge \overline{y}$ is intractable, almost all others are
  tractable), our results for MC will exhibit a rather diffuse borderline. This will become
  visible in the following overview and is addressed in the Conclusion. 
  Our most surprising
  intractability result is the \NP-hardness of the fragment that only allows the temporal operator
  \U\ and no propositional operator at all.
  Our most surprising tractability result
  is the \NL-completeness of MC for the fragment that only allows the temporal operators \F, \G, and
  the binary \OR-operator.
  Taking into account that MC for the fragment with only \F\ plus \AND\
  is already \NP-hard (which is a consequence from \cite{sicl85}), we would have expected
  the same lower bound for the ``dual'' fragment with only \G\ plus \OR, but in fact we show
  that even the fragment with \F\ and \G\ and \OR\ is tractable.
  In the presence of the \X-operator, the expected duality occurs: The fragment
  with \F, \X\ plus \AND\ and the one with \G, \X\ plus \OR\ are both \NP-hard.

  Table \ref{tab:overview} gives an overview of our results.
  The top row refers to the sets of Boolean operators given in Definition~\ref{def:clones}.
  These seven sets of Boolean operators are all relevant cases,
  which is due to Post's fundamental paper \cite{pos41} and Lemma \ref{lemma:const}.
  Entries in bold-face type denote completeness for the given complexity class
  under logspace reductions.
  (All reductions in this paper are logspace reductions $\redlogm$.)
  The entry \LS\ stands for logspace solvability.
  All other entries denote hardness results.
  Superscripts refer to the source of the corresponding result as explained in the legend.

  \newlength{\ab}\setlength{\ab}{-1pt}
  \newcommand{\stab}{\rule{0pt}{9.05pt}}
  \newcommand{\stabb}{\rule{0pt}{12.05pt}}

  \newcommand{\Stab}{\rule{0pt}{14pt}}

  \begin{table}[t]
    \renewcommand{\arraystretch}{1.3}
    \addtolength{\fboxsep}{-3pt}
    \centering
    \begin{small}
      \parbox{.64\textwidth}{%
        \begin{tabular}{@{~}l|@{\hspace{0.4pt}}*{7}{@{\hspace{-0.2pt}}c}@{~}}
          prop.\ operators & \cI         & \cN         & \cE         & \cV         & \cM         & \cL         & \cBF        \\[\ab]
          temp.\ operators &             &             &             &             &             &             &             \\
          \hline\stab
          \X               & $\WBNL{10}$ & $\WBNL{10}$ & $\WBNL{12}$ & $\WBNL{11}$ & $\GNP2$     & $\WBNL{14}$ & $\GNP{S}$   \\[\ab]
          \G               & $\WBNL{10}$ & $\WBNL{10}$ & $\WBNL{12}$ & $\WBNL{13}$ & $\GNP2$     &             & $\GNP{S}$   \\[\ab]
          \F               & $\WBNL{10}$ & $\WBNL{10}$ & $\GNP5$     & $\WBNL{11}$ & $\GNP2$     &             & $\GNP{S}$   \\[\ab]
          \F\G             & $\WBNL{10}$ & $\WBNL{10}$ & $\GNP{c}$   & $\WBNL{13}$ & $\GNP{c}$   &             & $\GNP{S}$   \\[\ab]
          \F\X             & $\WBNL{10}$ & $\WBNL{10}$ & $\GNP{c}$   & $\WBNL{11}$ & $\GNP{c}$   &             & $\GPS{T}$   \\[\ab]
          \G\X             & $\WBNL{10}$ & $\WBNL{10}$ & $\WBNL{12}$ & $\GNP6$     & $\GPS3$     &             & $\GPS{T}$   \\[\ab]
          \F\G\X           & $\WBNL{10}$ & $\WBNL{10}$ & $\GNP{c}$   & $\GNP{c}$   & $\GPS1$     &             & $\GPS{T}$   \\[\ab]
          \hline\stab
          \S               & $\WSLS{15}$ & $\WSLS{15}$ & $\WSLS{15}$ & $\WSLS{15}$ & $\WSLS{15}$ & $\WSLS{15}$ & $\WSLS{15}$ \\[\ab]
          \S\X             & $\GNP{8}$  & $\GNP{8}$  & $\GNP{8}$  & $\GNP{8}$  & $\GNP{8}$  & $\GNP{8}$  & $\GNP{8}$  \\[\ab]
          \S\G             & $\GNP8$     & $\GNP8$     & $\GNP8$     & $\GNP8$     & $\GPS4$     & $\GNP8$     & $\GPS4$     \\[\ab]
          \S\F             & $\WBNL{16}$ & $\GNP9$     & $\GNP9$     & $\WBNL{16}$ & $\GPS4$     & $\GNP9$     & $\GPS4$     \\[\ab]
          \S\F\G           & $\GNP{c}$   & $\GNP{c}$   & $\GNP{c}$   & $\GNP{c}$   & $\GPS{c}$   & $\GNP{c}$   & $\GPS{S}$   \\[\ab]
          \S\F\X           & $\GNP{c}$   & $\GNP{c}$   & $\GNP{c}$   & $\GNP{c}$   & $\GPS{c}$   & $\GNP{c}$   & $\GPS{T}$   \\[\ab]
          \S\G\X           & $\GNP{c}$   & $\GNP{c}$   & $\GNP{c}$   & $\GNP{c}$   & $\GPS{c}$   & $\GNP{c}$   & $\GPS{T}$   \\[\ab]
          \S\F\G\X         & $\GNP{c}$   & $\GNP{c}$   & $\GNP{c}$   & $\GNP{c}$   & $\GPS{c}$   & $\GNP{c}$   & $\GPS{T}$   \\[\ab]
          \hline\stab
          all other        & $\GNP7$     & $\GNP{c}$   & $\GNP{c}$   & $\GNP{c}$   & $\GPS3$     & $\GNP{c}$   & $\GPS{T}$   \\[-3pt]
          combinations     &             &             &             &             &             &             &             \\[-1pt]
          (\ie, with \U)  &             &             &             &             &             &             &
        \end{tabular}%
      }%
\hspace{-4mm}
      \framebox{%
      \parbox{.32\textwidth}{%
        \textbf{~~\rule{0mm}{5mm}Legend.}

        \medskip
        ~~(\PS\ stands for \PSPACE.)

        \medskip
        \renewcommand{\arraystretch}{1.15}
        ~~\begin{tabular}[t]{l@{~~}l@{}}
          1   &  Theorem~\ref{theorem:propositional negation} \ref{part:MC(F,G,X;M) PSPACE-h}                        \\
          2   &  Theorem~\ref{theorem:propositional negation} \ref{part:MC(F|G|X;M) NP-h}                            \\
          3   &  Theorem~\ref{theorem:propositional negation} \ref{part:MC(U|G,X;M) PSPACE-h}                        \\
          4   &  Theorem~\ref{theorem:propositional negation} \ref{part:MC(S,G|S,F;M) PSPACE-h}                      \\
          5   &  Corollary~\ref{cor:MC(F;E) NP-h}                                                             \\
          6   &  Theorem~\ref{theorem:MC(G,X;V) NP-h}                                                             \\
          7   &  Theorem~\ref{theorem:MC(U;.) NP-h}                                                               \\
          8   &  Theorem~\ref{theorem:MC(S,G;.) NP-h}                                                             \\
          9   &  Theorem~\ref{theorem:MC(S,F;E|N) NP-h}                                                           \\
          10  &  Theorem~\ref{theorem:MC(F,G,X;N) NL-c}                                                       \\
          11  &  Theorem~\ref{theorem:MC(F,X;V) MC(G,X;E) NL-c} \ref{part:FXV}                                \\
          12  &  Theorem~\ref{theorem:MC(F,X;V) MC(G,X;E) NL-c} \ref{part:GXE}                                \\
          13  &  Theorem~\ref{theorem:MC(F,G;V) NL-c}                                                         \\
          14  &  Theorem~\ref{theorem:MC(X;L) NL-c}                                                           \\
          15  &  Theorem~\ref{theorem:MC(S;.) in L}                                                           \\
          16  &  Theorem~\ref{theorem:MC(S,F;V) NL-c}                                                         \\
          $S$ &  Theorem~\ref{theorem:sicl85} \ref{part:MC(F;andorneg) NP-c}                                  \\
          $T$ &  Theorem~\ref{theorem:sicl85} \ref{part:MC(F,X|U|U,S,X;andorneg) PSPACE-c}                    \\[4pt]
          $c$ &  conclusion from                                                                              \\[-2pt]
              &  surrounding results
        \end{tabular}%
      }%
      }

    \end{small}
    \par\bigskip
    \renewcommand{\arraystretch}{1.5}
    \addtolength{\fboxsep}{3pt}
    \caption{An overview of complexity results for the model-checking problem}
    \label{tab:overview}
  \end{table}

  This paper is organised as follows. 
  Section \ref{sec:prelims} contains all necessary definitions and notation.
  In Section \ref{sec:NP}, we show \NP-hardness of all intractable cases,
  followed by Section \ref{sec:nl} with the \NL-completeness of almost all remaining cases. 
  We conclude in Section \ref{sec:concl}. 
\ifextended
\else
  Due to the limitations of space, we have left out a number of proofs.
  These can be found in the Technical Report \cite{bms+07}.
\fi

\section{Preliminaries}
\label{sec:prelims}

A \emph{Boolean function} is a function $f:\{0,1\}^n\rightarrow\{0,1\}$. 
We can identify an $n$-ary propositional function symbol $c$ with the $n$-ary Boolean function $f$ defined by: 
$f(a_1,\dots,a_n)=1$ if and only if the formula $c(x_1,\dots,x_n)$ becomes true 
when assigning $a_i$ to $x_i$ for all $1\leq i\leq n$.
An \emph{operator} is either a function or a function symbol, which becomes clear from the context.
Additionally to propositional operators we use the unary temporal operators 
$\X$ (next-time), $\F$ (eventually), $\G$ (invariantly) 
and the binary temporal operators $\U$ (until), and $\S$ (since).

Let $B$ be a finite set of Boolean operators and $T$ be a set of temporal operators. 
A \emph{temporal $B$-formula over $T$} is a formula $\varphi$ that is 
built from variables, propositional operators from $B$, and temporal operators from $T$.
  More formally, a temporal $B$-formula over $T$ is either a propositional variable or of the form $f(\varphi_1,\dots,\varphi_n)$
  or $g(\varphi_1,\dots,\varphi_m)$, where $\varphi_i$ are temporal $B$-formulae over $T$, $f$ is an $n$-ary
  propositional operator from $B$ and $g$ is an $m$-ary temporal operator from $T$. 
In \cite{sicl85}, complexity results for formulae using the temporal operators
  \F, \G, \X\ (unary), and \U, \S\ (binary) were presented. We extend these results to temporal $B$-formulae over subsets of those temporal operators.
  The set of variables appearing in $\varphi$ is denoted by $\VAR(\varphi).$ If $T=\{\X,\F,\G,\U,\S\}$ we call $\varphi$ a \emph{temporal
  $B$-formula}, and if $T=\emptyset$ we call $\varphi$ a \emph{propositional $B$-formula} or simply a \emph{$B$-formula}. The
  set of all temporal $B$-formulae over $T$ is denoted with $\L{T,B}.$

  A \emph{Kripke structure} is a triple $K = (W, R, \eta)$, where $W$ is a finite set of states,
  $R \subseteq W \times W$ is a total binary relation (meaning that, for each $a \in W$, there is some $b \in W$ such that $aRb$)\footnote{%
    In the strict sense, Kripke structures can have arbitrary binary relations.
    However, when referring to Kripke structures, we always assume their relations to be total.
  },
  and $\eta: W \to 2^{\VAR}$ for a set \VAR\ of variables.

  A model in linear temporal logic is a linear structure of states,
  which intuitively can be seen as different points of time, with propositional assignments.
  Formally, 
  a \emph{path} $p$ in $K$ is an infinite sequence denoted as $(p_0,p_1,\dots)$,
  where, for all $i \ge 0$, $p_i \in W$ and $p_iRp_{i+1}$.

  For a temporal $\{\wedge,\neg\}$-formula over $\{\F,\G,\X,\U,\S\}$ with variables from \VAR,
  a Kripke structure $K=(W,R,\eta)$, and a path $p$ in $K$,
  we define what it means that \emph{$\struct{p}{K}$ satisfies $\varphi$ in $p_i$}
  ($\struct{p}{K},i\vDash\varphi$):
  let $\varphi_1$ and $\varphi_2$ be temporal $\{\wedge,\neg\}$-formulae over $\{\F,\G,\X,\U,\S\}$ and let $x\in \VAR$ be a variable.

  \begin{center}
    \renewcommand{\arraystretch}{1.2}
    \begin{tabular}{lll}
      \multicolumn{3}{l}{$\struct{p}{K}, i\vDash 1$ ~ ~ and ~ ~ $\struct{p}{K}, i\not\vDash 0$} \\
      $\struct{p}{K},i\vDash x$                        &~~iff~~& $x\in\eta(p_i)$                                                        \\
      $\struct{p}{K},i\vDash \varphi_1\wedge\varphi_2$ &~~iff~~& $\struct{p}{K},i\vDash\varphi_1$ and $\struct{p}{K},i\vDash\varphi_2$  \\
      $\struct{p}{K},i\vDash \neg\varphi_1$            &~~iff~~& $\struct{p}{K},i \nvDash\varphi_1$                                     \\
      $\struct{p}{K},i\vDash \F\varphi_1$              &~~iff~~& there is a $j\geq i$ such that $\struct{p}{K},j \vDash\varphi_1$       \\
      $\struct{p}{K},i\vDash \G\varphi_1$              &~~iff~~& for all $j\geq i$, $\struct{p}{K},j \vDash\varphi_1$                   \\
      $\struct{p}{K},i\vDash \X\varphi_1$              &~~iff~~& $\struct{p}{K},i+1 \vDash\varphi_1$                                    \\
      $\struct{p}{K},i\vDash \varphi_1\U\varphi_2$     &~~iff~~& there is an $\ell\geq i$ such that $\struct{p}{K},\ell\vDash\varphi_2$,\\[-3pt]
                                                       &       & and for every $i\leq j<\ell$,~ $\struct{p}{K},j \vDash\varphi_1$       \\
      $\struct{p}{K},i\vDash \varphi_1\S\varphi_2$     &~~iff~~& there is an $\ell\leq i$ such that $\struct{p}{K},\ell\vDash\varphi_2$,\\[-3pt]
                                                       &       & and for every $\ell<j\leq i$,~ $\struct{p}{K},j \vDash\varphi_1$
    \end{tabular}
    \renewcommand{\arraystretch}{1.5}
  \end{center}

  Since
  every Boolean operator can be composed from $\wedge$ and $\neg$, the above definition generalises to temporal
  $B$-formulae for arbitrary sets $B$ of Boolean operators.

  This paper examines the model-checking problems \mc{T,B} for finite sets $B$ of Boolean functions
and sets $T$ of temporal operators.

  \dproblem{\mc{T,B}}{$\mcinst{\varphi,K,a}$, where $\varphi \in \L{T,B}$ is a formula, $K = (W, R, \eta)$ is a Kripke structure,
      and $a \in W$ is a state}{Is there a path $p$ in $K$ such that $p_0 = a$ and $\struct{p}{K},0 \vDash \varphi$?}

  Sistla and Clarke \cite{sicl85} have established the computational complexity of the model-checking problem
  for temporal $\{\wedge,\vee,\neg\}$-formulae over some sets of temporal operators.

  \begin{theorem}[\cite{sicl85}]\label{theorem:sicl85}
    ~\par
    \begin{Enum}
      \item\label{part:MC(F;andorneg) NP-c}
            $\mc{\{\F\},\{\wedge,\vee,\neg\}}$ is \NP-complete.
      \item\label{part:MC(F,X|U|U,S,X;andorneg) PSPACE-c}
            $\mc{\{\F,\X\},\{\wedge,\vee,\neg\}}$,
            $\mc{\{\U\},\{\wedge,\vee,\neg\}}$, and
            $\mc{\{\U,\S,\X\},\{\wedge,\vee,\neg\}}$ are \linebreak
            \PSPACE-complete.
    \end{Enum}
  \end{theorem}

  Since there are infinitely many finite sets of Boolean functions,
  we introduce some algebraic tools to classify the complexity of the infinitely many arising satisfiability problems.
  We denote with \proj{n}{k} the $n$-ary projection to the $k$-th variable, where $1\leq k\leq n$, \ie, $\proj{n}{k}(x_1,\dots,x_n)=x_k$,
  and with \const{n}{a} the $n$-ary constant function defined by $\const{n}{a}(x_1,\dots,x_n)=a$.
  For $\const{1}{1}(x)$ and $\const{1}{0}(x)$ we simply write 1 and 0.
  A set $C$ of Boolean functions is called a \emph{clone} if it is closed under superposition,
  which means $C$ contains all projections and $C$ is closed under arbitrary composition \cite{pip97b}.
  For a set $B$ of Boolean functions we denote with $\clone{B}$ the smallest clone containing $B$
  and call $B$ a \emph{base} for $\clone{B}$.
  In \cite{pos41} Post classified the lattice of all clones
  and found a finite base for each clone.

  The definitions of all clones as well as the full inclusion graph can be found, for example, in~\cite{bcrv03}.
\ifextended
  The following lemma
\else
  The following lemma, which we prove in \cite{bms+07}, 
\fi
  implies that only clones with both constants $0,1$ are relevant for the model-checking problem;
  hence we will only define those clones. Note, however, that our results will carry over to all clones.

\begin{lemma}\label{lemma:const}
Let $B$ be a finite set of Boolean functions and $T$ be a set of temporal operators. 
Then $\dt{T,B\cup\{0,1\}}\equiv^{\log}_m\dt{T,B}$.
\end{lemma}

\ifextended
  \begin{proof}
  $\dt{T,B}\redlogm\dt{T,B\cup\{0,1\}}$ is trivial. For $\dt{T,B\cup\{0,1\}}\redlogm\dt{T,B}$ let $\langle \varphi,K,a\rangle$
  be an instance of \dt{T,B\cup\{0,1\}} for a Kripke structure $K=(W,R,\eta)$ and let $\perp$ and $\top$ be two fresh variables.
  We define a new Kripke structure $K'=(W,R,\eta')$ 
  where $\eta'(\alpha)=\eta(\alpha)\cup\{\top\}$ and we define $\varphi'$ to be a copy of $\varphi$ where
  every appearance of 0 is replaced by $\perp$ and every appearance of 1 by $\top$. It holds that $\langle \varphi',K',a\rangle$ is an instance of
  \dt{T,B} and that $\langle \varphi,K,a\rangle\in\dt{T,B\cup\{0,1\}}$ if and only if $\langle \varphi',K',a\rangle\in\dt{T,B}$.
  \end{proof}
\fi

  Because of Lemma \ref{lemma:const} it is sufficient to look only at the clones with constants,
  which are introduced in Definition~\ref{def:clones}. Their bases and inclusion structure
  are given in Figure~\ref{figure:clones with constants}.

  \begin{figure}
    \renewcommand{\arraystretch}{1.1}
    \centering
    \parbox{.3\textwidth}{%
      \begin{tabular}{l|l}
        clone & base \\
        \hline\stabb
        \cBF  & $\{\wedge,\neg\}$        \\
        \cM   & $\{\vee,\wedge,0,1\}$    \\
        \cL   & $\{\oplus,1\}$           \\
        \cV   & $\{\vee,1,0\}$           \\
        \cE   & $\{\wedge,1,0\}$         \\
        \cN   & $\{\neg,1,0\}$           \\
        \cI   & $\{0,1\}$
      \end{tabular}
    }
    \hspace*{1cm}
    \parbox{.4\textwidth}{%
      \includegraphics{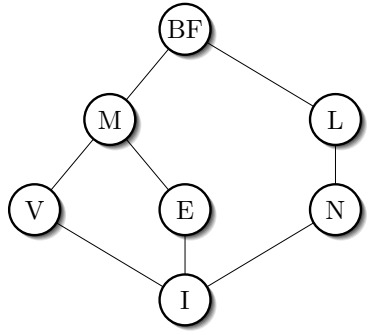}
    }
    \renewcommand{\arraystretch}{1.5}
    \caption{Clones with constants}
    \label{figure:clones with constants}
  \end{figure}

  \begin{definition}
    \label{def:clones}
    Let $\oplus$ denote the binary exclusive or. Let $f$ be an $n$-ary Boolean function.
    \begin{Enum}
      \item
        $\cBF$ is the set of all Boolean functions.
      \item
        $\cM$ is the set of all monotone functions, that is,
        the set of all functions $f$ where $a_1\leq b_1$, $\dots$, $a_n\leq b_n$ implies $f(a_1,\dots,a_n)\leq f(b_1,\dots,b_n)$.
      \item
        $\cL$ is the set of all linear functions, that is,
        the set of all functions $f$ that satisfy $f(x_1,\dots,x_n) = c_0 \oplus (c_1 \wedge x_1) \oplus \dots \oplus (c_n \wedge x_n)$,
        for constants $c_i$.
      \item
        $\cV$ is the set of all functions $f$ where $f(x_1,\dots,x_n) = c_0 \vee (c_1\wedge x_1) \vee \dots \vee (c_n\wedge x_n)$,
        for constants $c_i$.
      \item
        $\cE$ is the set of all functions $f$ where $f(x_1,\dots,x_n) = c_0 \wedge(c_1\vee x_1)\wedge\dots\wedge(c_n\vee x_n)$,
        for constants $c_i$.
      \item
        $\cN$ is the set of all functions that depend on at most one variable.
      \item
        $\cI$ is the set of all projections and constants.
    \end{Enum}
  \end{definition}

  There is a strong connection between propositional formulae and Post's lattice.
  If we interpret propositional formulae as Boolean functions,
  it is obvious that $[B]$ includes exactly those functions that can be represented by $B$-formulae.
  This connection has been used various times to classify the complexity of problems related to propositional formulae.
  For example, Lewis presented a dichotomy for the satisfiability problem for propositional $B$-formulae:
  it is \NP-complete if $x\wedge\overline{y}\in\clone{B}$, 
  and solvable in \PTIME\ otherwise \cite{lew79}. Furthermore,
  Post's lattice has been applied to the equivalence problem \cite{rei01},
  to counting \cite{rewa99-dt} and finding minimal \cite{revo03} solutions,
  and to learnability \cite{dal00} for Boolean formulae.
  The technique has been used in non-classical logic as well:
  Bauland et al. achieved a trichotomy in the context of modal logic,
  which says that the satisfiability problem for modal formulae is, depending on the allowed propositional connectives,
  \PSPACE-complete, \CONP-complete, or solvable in \PTIME\ \cite{bhss06}.
  For the inference problem for propositional circumscription, Nordh presented another trichotomy theorem \cite{nor05}.

  An important tool in restricting the length of the resulting formula in many of our reductions is the
\ifextended
  following lemma.
\else
  following lemma, which we prove in \cite{bms+07}.
\fi

\begin{lemma}\label{lemma:subred}
Let $B\subseteq\{\wedge,\vee,\neg\}$, and
let $C$ be a finite set of Boolean functions such that $B\subseteq[C]$.
Then $\dt{T,B}\redlogm\dt{T,C}$ for every set $T$ of temporal operators.
\end{lemma}

\ifextended
 \begin{proof}
  Let $D = C \cup \{0,1\}$.
  From Lemmas 1.4.4 and 1.4.5 in~\cite{sch07} we directly conclude:
  Let $f$ be one of the functions \OR, \AND, and \NOT\ such that $f\in\clone D$. Let $k$ be the arity of $f$. 
  Then there is a $D$-formula $\varphi(x_1,\dots,x_k)$ representing $f$, such that every variable occurs only once in $\varphi$. 
  Hence $\dt{T,B}\redlogm\dt{T,C\cup\{0,1\}}$.
  From Lemma~\ref{lemma:const} follows $\dt{T,C\cup\{0,1\}}\redlogm\dt{T,C}$.
 \end{proof}
\fi

It is essential for this Lemma that $B\subseteq\{\wedge,\vee,\neg\}$.
For, \eg, $B=\{\oplus\}$,
it is open whether $\dt{T,B}\redlogm\dt{T,\cBF}$.
This is a reason why we cannot immediately transform
upper bounds proven by Sistla and Clarke~\cite{sicl85}---%
for example, $\dt{\{\F,\X\},\{\wedge,\vee,\neg\}}\in\PSPACE$---%
to upper bounds for all finite sets of Boolean functions---%
\ie, it is open whether for all finite sets $B$ of Boolean functions,
$\dt{\{\F,\X\},B}\in\PSPACE$.

\newcommand{\cA}{\mathcal{A}}

\section{The bad fragments: intractability results}
\label{sec:NP}

Sistla and Clarke \cite{sicl85} and Markey \cite{mar04} have considered
the complexity of model-checking for temporal $\{\wedge,\vee,\neg\}$-formulae
restricted to atomic negation and propositional negation, respectively.
\ifextended
  We define a temporal $B$-formula with \textit{propositional negation} to be a temporal $B$-formula
  where additional negations are allowed, but only in such a way
  that no temporal operator appears in the scope of a negation sign.
  In the case that negation is an element of $B$,
  a temporal $B$-formula with propositional negation is simply a temporal $B$-formula.
  In \cite{sicl85}, \textit{atomic negation} is considered,
  which restricts the use of negation even further---negation is only allowed directly for variables.
  We will now show that propositional negation does not make any difference for the complexity of the model checking problem.
  Since this obviously implies that atomic negation inherits the same complexity behaviour,
  we will only speak about propositional negation in the following.
  The proof of the following lemma is similar to that of Lemma~\ref{lemma:const}.

  \begin{lemma}\label{lemma:propositional negation}
  Let $T$ be a set of temporal operators, and $B$ a finite set of Boolean functions.
  We use $\dtp{T,B}$ to denote the model-checking problem 
  $\dt{T,B}$ extended to $B$-formulae with propositional negation.
  Then $\dtp{T,B}\equiv^{\log}_m\dt{T,B}$.
  \end{lemma}
\begin{proof}
The reduction $\dt{T,B}\redlogm\dtp{T,B}$ is trivial. For $\dtp{T,B}\redlogm\dt{T,B}$, assume that negation is not an element of $B$, otherwise there is nothing to prove. Let $\mcinst{\varphi, K, a}$ be an instance of \dtp{T,B}, where $K=(W,R,\eta)$.
Let $x_1,\ldots,x_m$ be the variables that appear in $\varphi$, and for each formula of the kind $\neg\psi(x_1,\dots,x_n)$ appearing in $\varphi,$ let $y_{\neg\psi}$ be a new variable. Note that since only propositional negation is allowed in $\varphi,$ in these cases $\psi$ is purely propositional.

We obtain $K'=(W,R,\eta')$ from $K$ by extending $\eta$
to the variables $y_{\neg\psi}$ in such a way that $y_{\neg\psi}$ is \textit{true} in a state if and only if $\psi(x_1,\dots,x_n)$ is \textit{false}.
Finally,
to obtain  $\varphi'$ from $\varphi$ 
we replace every appearance of $\neg\psi(x_1,\dots,x_n)$ with $y_{\neg\psi}.$
Now, $\varphi'$ is a temporal $B$-formula.
By the construction 
it is straightforward to see that
$\langle \varphi, K, a \rangle \in \dtp{T,B}$ iff
$\langle \varphi', K', a \rangle \in \dt{T,B}$.
\end{proof}

\else
  Propositional negation does not affect
  the complexity of the model checking problem.
  This can be proven similarly to Lemma~\ref{lemma:const}.
\fi
Using Lemma~\ref{lemma:subred} in addition,
we can generalise the above mentioned hardness results from \cite{sicl85,mar04}
for temporal monotone formulae
to obtain the following intractability results for model-checking.
\ifextended
\else
  The proofs of all results in this section are given in \cite{bms+07}.
\fi

\begin{theorem}\label{theorem:propositional negation}
   Let $M_{+}$ be a finite set of Boolean functions 
   such that $\cM\subseteq\clone{M_{+}}$.
   Then 
    \begin{Enum}
      \item\label{part:MC(F,G,X;M) PSPACE-h}
            \dt{\{\F,\G,\X\},M_{+}} is \PSPACE-hard.
      \item\label{part:MC(F|G|X;M) NP-h}
            $\mc{\{\F\},M_{+}}$, $\mc{\{\G\},M_{+}}$, 
            and $\mc{\{\X\},M_{+}}$ are \NP-hard.
      \item\label{part:MC(U|G,X;M) PSPACE-h}
            $\mc{\{\U\},M_{+}}$ and $\mc{\{\G,\X\},M_{+}}$ are \PSPACE-hard.
      \item\label{part:MC(S,G|S,F;M) PSPACE-h}
            $\mc{\{\S,\G\},M_{+}}$ and $\mc{\{\S,\F\},M_{+}}$ are \PSPACE-hard.
    \end{Enum}
\end{theorem}

In Theorem 3.5 in \cite{sicl85} it is shown that  \dt{\{\F\},\set{\wedge,\vee,\neg}} is \NP-hard.
In fact, Sistla and Clarke give a reduction from \threesat\ to   \dt{\{\F\},\set{\wedge}}.
The result for arbitrary bases $B$ generating a clone above $\cE$ 
follows from Lemma~\ref{lemma:subred}.

\begin{corollary}\label{cor:MC(F;E) NP-h}
Let $E_{+}$ be a finite set of Boolean functions such that $\cE\subseteq\clone{E_{+}}$.
Then \dt{\{\F\},E_{+}} is \NP-hard.
\end{corollary}

The model-checking problem 
for temporal $\{\G,\X\}$-$\{\wedge,\vee\}$-formulae
is \PSPACE-complete 
(Theorem~\ref{theorem:propositional negation}\ref{part:MC(U|G,X;M) PSPACE-h} due to \cite{mar04}).
The Boolean operators $\{\wedge,\vee\}$ are a basis of $\cM$,
the class of monotone Boolean formulae.
What happens for fragments of \cM ?
In Theorem~\ref{theorem:MC(F,X;V) MC(G,X;E) NL-c} we will show
that \dt{\{\G,\X\},\cE} is \NL-complete,
\ie, the model-checking problem for
temporal $\{\wedge\}$-formulae over $\{\G,\X\}$ is very simple.
We can prove that switching from $\wedge$ to $\vee$ makes the problem intractable.
As notation, we use $\LIT(\varphi)$ to denote the literals obtained from variables that appear in $\varphi$.

\begin{theorem}
\label{theorem:MC(G,X;V) NP-h}
Let $V_{+}$ be a finite set of Boolean functions such that $\cV\subseteq\clone{V_{+}}$.
Then $\dt{\{\G,\X\},V_{+}}$ is \NP-hard.
\end{theorem}

\ifextended
  \begin{proof}
  It suffices to give a reduction from \threesat\  to $\dt{\{\G,\X\},\{\vee\}}$
  (due to Lemma~\ref{lemma:subred}).
  A formula $\psi$ in 3CNF is mapped to an instance $\mcinst{\psi',K(\psi),q_1}$ 
  of \dt{\{\G,\X\},\{\vee\}} as follows.
  Let $\psi=C_1\wedge\ldots\wedge C_m$ consist of $m$ clauses,
  and $n=|\VAR(\psi)|$ variables.
  The Kripke structure $K(\psi)$
  has states $Q=\{q_1,\ldots,q_m\}$ containing
  one state for every clause,
  a sequence of states 
  $P=\{l^j \mid l\in \LIT(\psi), 0\leq j\leq m-1\}$ for every literal,
  and a final sink state $z$.
  That is, the set of states is $Q\cup P \cup \{z\}$.
  The variables of $K(\psi)$ are $b_1,\ldots,b_m,c$.
  Variable $b_a$ is assigned \emph{true} in a state $l_i^0$ iff
  literal $l_i$ is contained in clause $C_a$.
  In all other states, every $b_i$ is \textit{false}.
  Variable $c$ is assigned \textit{true} in all states in $P \cup \{z\}$.

  The relation between the states is $E_1\cup E_2\cup E_3\cup E_4\cup E_5$ as follows.
  It starts with the path $q_1,\ldots,q_m$:
  $E_1=\{(q_i,q_{i+1})\mid i=1,2,\ldots,m-1\}$.
  $q_m$ has an edge to $x^0_1$ and an edge to $\overline{x}^0_1$:
  $E_2 = \{(q_m,x^0_1),(q_m,\overline{x}^0_1)\}$.
  Each $l_i^0$ is the starting point of a path
  $l_i^0,l_i^1,\ldots,l_i^{m-1}$:
  $E_3=\{(l_i^j,l_i^{j+1})\mid l_i\in\LIT(\psi), j=0,1,\ldots,m-2\}$.
  Each endpoint of these paths has
  both the literals with the next index resp. the final sink state as neighbours:
  $E_4=\{(l_i^{m-1},l_{i+1}^0) \mid i=1,2,\ldots,n-1, l_i\in\LIT(\psi)\}
  \cup \{ (x_n^{m-1},z),(\overline{x}_n^{m-1},z)\}$.
  The final sink state $z$ has an edge to $z$ itself, $E_5=\{(z,z)\}$.

  \begin{figure}[ht]
  {

  \centering

  \includegraphics[scale=0.8]{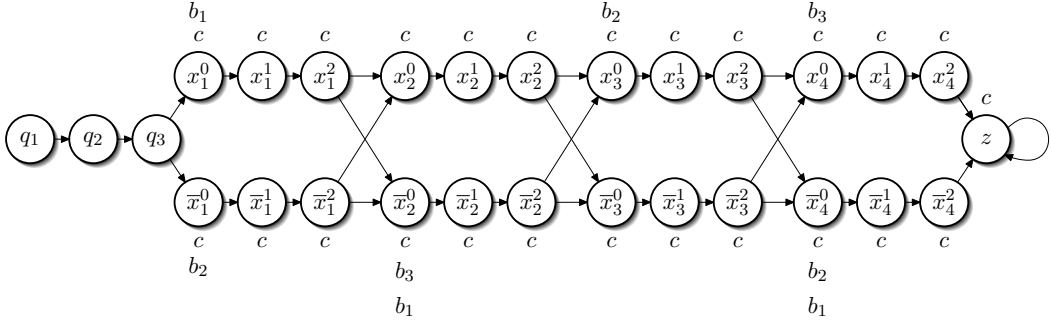}

  }

  \caption{The Kripke structure $K(\psi_0)$ for
  $\psi_0=(x_1\vee \neg x_2 \vee \neg x_4) \wedge (\neg x_1\vee x_3 \vee \neg x_4) \wedge (\neg x_2 \vee x_4)$.
  }
  \label{figure:structure for GXV}
  \end{figure}

  Figure~\ref{figure:structure for GXV} shows an example for a formula $\psi_0$ 
  and the Kripke structure $K(\psi_0)$.
  Notice that every path in such a Kripke structure $K(\psi)$ corresponds to an assignment
  to the variables in $\psi$.
  A path corresponds to a satisfying assignment iff
  for every $b_i$ the path contains a state that $b_i$ is assigned to.
  We are now going to construct a formula $\psi'$ to express this property.
  If we were allowed to use the $\wedge$ in $\psi'$, this would be easy.
  But, the formula $\psi'$ consists only of operators \G, \X, $\vee$,
  and of variables $b_1,\ldots,b_m,c$.
  In order to define $\psi'$, 
  we use formulae $\varphi_i$ and $\psi'_i$ defined as follows.
  For $i=1,2,\ldots,m$ define 
  $$\varphi_i= \bigvee\limits_{k=1,2,\ldots,n} \X^{k\cdot m-(i-1)} b_i ~ .$$
  Intuitively, $\varphi_i$ says that 
  $b_i$ is satisfied in a state in distance $d$,
  where $d \equiv m-(i-1) \pmod m$.
  The state $q_j$ is the only state in $Q$
  where $\varphi_j$ can hold.
  Every path $p$ in $K(\psi)$ has the form
  $p=(q_1,q_2,\ldots,q_m,l_1^0,\ldots,l_n^{m-1}, z, z,\ldots)$.
  Every state except for $z$ appears at most once in $p$.
  For the sake of simplicity,
  we use the notation $\struct{p}{K(\psi)},q_i \vDash \alpha$ 
  for $\struct{p}{K(\psi)},i-1\vDash\alpha$ (for $i=1,2,\ldots,m$),
  and $\struct{p}{K(\psi)},l_i^j \vDash \alpha$ 
  for $\struct{p}{K(\psi)},m+(i-1)\cdot m+j\vDash\alpha$.
  We use for a path $p=(q_1,q_2,\ldots)$ in $K(\psi)$ and $1\leq i\leq m$
  the notation $\struct{p}{K(\psi)},q_i \vDash \alpha$ for $\struct{p}{K(\psi)},i-1\vDash\alpha$.

  \begin{Claim}\label{claim:np1}
  For every path $p$ in $K(\psi)$ and $1\leq i,j\leq m$ holds:\
  If $\struct{p}{K(\psi)}, q_i \vDash \varphi_j$, then $i=j$.
  \end{Claim}
  
  \begin{Proofofclaim}[\ref{claim:np1}]
  Assume $\struct{p}{K(\psi)}, q_i \vDash \varphi_j$, where $1\leq i,j\leq m$.
  By the definition of $\varphi_j$,
  it follows that $\struct{p}{K(\psi)}, (j-1)+k\cdot m-(i-1) \vDash b_j$ for some $k$ with $1\leq k\leq n$.
  Consider any path $p$ in $K(\psi)$.
  After the initial part $(p_0,\ldots,p_{m-1})=(q_1,\ldots,q_m)$ of $p$ follows a sequence
  $(p_{m},\ldots,p_{m\cdot(n+1)})$
  of $n\cdot m$ states,
  where $p_i=l_{\lfloor i/m\rfloor}^{i \bmod m}$ (for $i=m,\ldots,m\cdot(n+1)$).
  Therefore, $p_{(j-1)+k\cdot m-(i-1)}=l_r^{(j+k\cdot m-i) \bmod m} = l_r^{(j-i) \bmod m}$
  for some $r$ (that does not matter here).
  But $\struct{p}{K(\psi)}, l_r^w \vDash b_j$ implies $w=0$, by the definition of $K(\psi)$,
  and therefore $(j-i) \bmod m = 0$.
  Since $1\leq i,j\leq m$, it follows that $j=i$.
  \end{Proofofclaim}
  
  The formulae $\psi'_i$ are defined inductively for $i=m+1,\ldots,2,1$ as follows (as before, we can use $\vee$ in our construction):
  $$\psi'_{m+1} = c \mbox{~ ~ ~ and ~ ~ ~}
     \psi'_{i} = \G\left(\varphi_i \vee \G \psi'_{i+1}\right) ~
  (\mbox{for } m\geq i \geq 1) ~ .
  $$
  Finally, $\psi'=\psi'_1$.

  It is clear that the reduction function $\psi \mapsto \mcinst{\psi',K(\psi),q_1}$ 
  can be computed in logarithmic space.
  It remains to prove the correctness of the reduction.
  Using Claim~\ref{claim:np1}, we make the following observation.

  \begin{Claim}\label{claim:np2}
  For every path $p=(q_1,q_2,\ldots)$ in $K(\psi)$
  and $i=1,2,\ldots, m$ holds:~ ~ 

  \centerline{$\struct{p}{K(\psi)},q_i \vDash \psi'_i$
  ~ if and only if ~ 
  for $j=i,i+1,\ldots,m$ holds $\struct{p}{K(\psi)},q_j \vDash \varphi_j$.}
  \end{Claim}

  \begin{Proofofclaim}[\ref{claim:np2}]
  The direction from right to left is straightforward.
  To prove the other direction, we use induction.

  As base case we consider $i=m$.
  Assume $\struct{p}{K(\psi)},q_m\vDash\G(\varphi_m \vee \G c)$.
  By construction of $K(\psi)$ holds $\struct{p}{K(\psi)},q_m\nvDash c$,
  and therefore $\struct{p}{K(\psi)},q_m\vDash \varphi_m$ holds.

  For the inductive step,
  assume $\struct{p}{K(\psi)},q_i \vDash \G(\varphi_i\vee\G\psi'_{i+1})$.
  Claim~\ref{claim:np1} proves $\struct{p}{K(\psi)},q_i \nvDash \varphi_j$ for $j\not=i$,
  and with $\struct{p}{K(\psi)},q_i\nvDash c$ we obtain $\struct{p}{K(\psi)},q_i \nvDash \G\psi'_{i+1}$.
  This implies $\struct{p}{K(\psi)},q_i \vDash \varphi_i$ and  $\struct{p}{K(\psi)},q_{i+1} \vDash \psi'_{i+1}$.
  By the inductive hypothesis, the claim follows.
  \end{Proofofclaim}
  
  For a path $p$ in $K(\psi)$, let $\cA_p$ be the corresponding assignment for $\psi$.
  It is clear that $\struct{p}{K(\psi)},q_i\vDash\varphi_i$ if and only if $\cA_p$ satisfies clause $C_i$ of $\psi$.
  Using Claim~\ref{claim:np2},
  it follows that  $\struct{p}{K(\psi)},q_1\vDash\psi'$ if and only if $\cA_p$ satisfies all clauses of $\psi$,
  \ie, $\cA_p$ satisfies $\psi$.
  Using the one-to-one correspondence between paths in $K(\psi)$ and
  assignments to the variables of $\psi$ we get
  $\psi\in \threesat$ if and only if $\mcinst{\psi', K(\psi), q_1}\in\dt{\{\G,\X\}, \{\vee\}}$.
  \end{proof}
\fi

From \cite{sicl85} it follows that $\dt{\{\G,\X\},\cV}$ is in \PSPACE.
It remains open whether $\dt{\{\G,\X\},\cV}$ or $\dt{\{\G,\X\},\cM}$
have an upper bound below \PSPACE.

Next,
we consider formulae with the until-operator or the since-operator.
We first show that using the until-operator makes model-checking intractable.

\begin{theorem}\label{theorem:MC(U;.) NP-h}
Let $B$ be a finite set of Boolean functions.
Then \dt{\{\U\},B} is \NP-hard.
\end{theorem}

\ifextended
  \begin{proof}
  We give a reduction from \threesat\ to $\dt{\{\U\},\emptyset}$.
  This means, that we do not need any Boolean operators in the temporal formula over $\{\U\}$
  to which a \threesat\ instance is mapped.
  Let $\psi= C_1 \wedge C_2 \wedge \ldots \wedge C_m$ 
  be a 3CNF formula consisting of $m$ clauses and $n$ variables.
  The structure $K(\psi)$ has states 
  $\{q_1,\ldots,q_m\} \cup \LIT(\psi) \cup\{s\}$,
  with initial state $q_1$.
  The assignment for state $q_i$ is $\{a_1,\ldots,a_i\}$ (for $i=1,2,\ldots,m$),
  and for state $l_i$ it is $\{a_1,\ldots,a_m\}\cup\{b_j\mid 
  \mbox{Literal $l_i\in\LIT(\psi)$ appears in clause $C_j$}\}$.
  In state $s$, no variable is assigned true.
  The relation between the states is as follows.
  Each $q_i$ ($i=1,2,\ldots,m-1$) has an edge to $q_{i+1}$,
  $q_m$ has edges to $x_1$ and to $\overline{x}_1$,
  each $l_i$ ($i=1,2,\ldots,n-1$) has edges to $x_{i+1}$ and to $\overline{x}_{i+1}$,
  and $x_n$ and $\overline{x}_n$ have an edge to $s$.
  $s$ has an edge to $s$ only.
  Figure~\ref{figure:structure for until}
  gives an example.
  \begin{figure}[h]
  {

  \centering

  \includegraphics[scale=0.8]{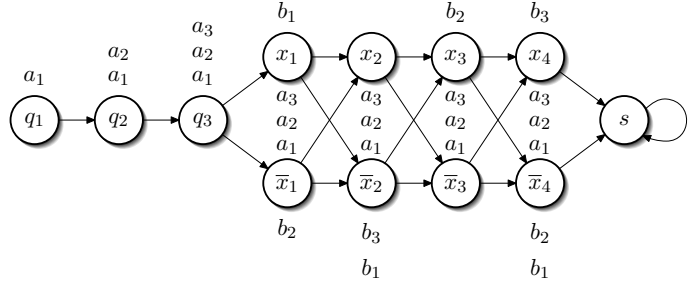}

  }

  \caption{Structure $K(\psi)$ for 
  $\psi=(x_1\vee \neg x_2 \vee \neg x_4) \wedge (\neg x_1\vee x_3 \vee \neg x_4) \wedge (\neg x_2 \vee x_4)$}
  \label{figure:structure for until}
  \end{figure}
  The following facts are easy to verify for any path $p$ in $K(\psi)$.
  For the sake of simplicity,
  we use for a path $p=(q_1,q_2,\ldots)$ in $K(\psi)$ and $1\leq i\leq m$
  the notation $\struct{p}{K(\psi)},q_i \vDash \alpha$ for $\struct{p}{K(\psi)},i-1\vDash\alpha$.
  \begin{description}
  \item[Fact 1] For $1\leq j < i \leq m$ holds: $\struct{p}{K(\psi)}, q_j \nvDash a_i \U b_i$ .
  \item[Fact 2] For $1\leq i \leq m$ holds: \ $\exists t: \struct{p}{K(\psi)},t\vDash a_i\U b_i$ iff $\struct{p}{K(\psi)}, q_i \vDash a_i \U b_i$ .
  \end{description}
  The formulae $\varphi_0, \varphi_1, \ldots$
  are defined inductively as follows.
  $$\varphi_0 ~ = ~ 1 \mbox{~ ~ ~ and ~ ~ ~} 
  \varphi_{i+1} ~ = ~ \varphi_i \U \left(a_{i+1} \U b_{i+1}\right).$$
  The reduction from $\threesat$ to \dt{\{\U\},\emptyset} is 
  the mapping $\psi \mapsto (\varphi_m,K(\psi),q_1)$,
  where $\psi$ is a 3CNF-formula with $m$ clauses.
  This reduction can evidently be performed in logarithmic space.
  To prove its correctness,
  we use the following claim.

  \begin{Claim}\label{claim:np3}
  Let $K(\psi)$ be constructed  from a formula $\psi$ with $m$ clauses,
  and let $p$ be a path in $K(\psi)$.
  For $j=1,2,\ldots,m$, it holds that  
  $\struct{p}{K(\psi)},q_1\vDash \varphi_j ~\mbox{if and only if}$
         $\struct{p}{K(\psi)},q_1\vDash \varphi_{j-1} \mbox{ and } \struct{p}{K(\psi)},q_j \vDash a_j \U b_j~ .
  $
  \end{Claim}
  
  \begin{Proofofclaim}[\ref{claim:np3}]
  We prove the claim by induction.
  The base case $j=1$ is straightforward:
  $\struct{p}{K(\psi)},q_1 \vDash \,1\, \U (a_1 \U b_1)$ is equivalent to $\exists t: \struct{p}{K(\psi)},t \vDash (a_1 \U b_1)$
  which by Fact 2 is equivalent to $\struct{p}{K(\psi)},q_1 \vDash a_1 \U b_1$.
  The inductive step is split into two cases.
  First, assume $\struct{p}{K(\psi)},q_1\vDash \varphi_{j+1}$.
  Since $\varphi_{j+1} = \varphi_j \U (a_{j+1} \U b_{j+1})$,
  it follows that $\exists t: \struct{p}{K(\psi)},t \vDash a_{j+1} \U b_{j+1}$.
  Using Fact 2, we conclude $\struct{p}{K(\psi)},q_{j+1} \vDash a_{j+1} \U b_{j+1}$.
  By Fact 1, $\struct{p}{K(\psi)},q_{1} \nvDash a_{j+1} \U b_{j+1}$.
  By the initial assumption, this leads to $\struct{p}{K(\psi)},q_1\vDash \varphi_{j}$.
  Second, assume 
  $\struct{p}{K(\psi)},q_1\vDash \varphi_{j}$ and $\struct{p}{K(\psi)},q_{j+1} \vDash a_{j+1} \U b_{j+1}$.
  Using the induction hypothesis,
  we obtain $\struct{p}{K(\psi)},q_i\vDash a_i \U b_i $ for $i=1,2,\ldots,j+1$.
  By the construction of $\varphi_{j+1}$ we immediately get
  $\struct{p}{K(\psi)},q_1 \vDash \varphi_{j+1}$.
  \end{Proofofclaim}

  We have a one-to-one correspondence
  between paths in $K(\psi)$ and assignments
  to variables of $\psi$.
  For a path $p$ we will denote the corresponding assignment by $\cA_p$.
  Using Claim~\ref{claim:np3}, it is easy to see that the following properties are equivalent.
  \begin{enumerate}
  \item 
   $\cA_p$ is a satisfying assignment for $\psi$.
  \item
  Path $p$ in $K(\psi)$ contains for every $i=1,2,\ldots,m$
  a state with assignment $b_i$.
  \item 
  $\struct{p}{K(\psi)},q_i\vDash a_i \U b_i$ for $i=1,2,\ldots,m$.
  \item
  $\struct{p}{K(\psi)},q_1\vDash \varphi_m$.
  \end{enumerate}
  This concludes the proof
  that $\psi\in\threesat$ if and only if $\mcinst{\varphi_m,K(\psi),q_1}\in \dt{\{\U\},\emptyset}$.
  \end{proof}
\fi


Although the until-operator and the since-operator appear to be similar,
model-checking for formulae that use the since-operator as only operator
is as simple as for formulae without temporal operators---see Theorem~\ref{theorem:MC(S;BF) in L}.
The reason is that the since-operator has no use
at the beginning of a path of states, where no past exists.
It needs other temporal 
operators that are able
to enforce to visit a state on a path that has a past.


\begin{theorem}\label{theorem:MC(S,G;.) NP-h}\label{theorem:MC(S,X;.) NP-h}
Let $B$ be a finite set of Boolean functions. 
Then \dt{\{\X,\S\},B} and \dt{\{\G,\S\},B} are \NP-hard.
\end{theorem}

\ifextended
  \begin{proof}
  We give a reduction from \threesat\ to \dt{\{\G,\S\},\emptyset} that is similar to that
  in the proof of Theorem~\ref{theorem:MC(U;.) NP-h} for \dt{\{\U\},\emptyset}.
  Let $\psi$ be an instance of \threesat,
  and let $K(\psi)$ be the structure as in the proof of Theorem~\ref{theorem:MC(U;.) NP-h}.
  From $K(\psi)=(W,R,\eta)$ we obtain the structure $H(\varphi)=(W',R',\eta')$ as follows.
  First, we add a new state $t$, \ie, $W'=W\cup\{t\}$.
  Second, replace $R$ by its inverse $R^{-1}=\{(v,u)\mid (u,v)\in R\}$
  from which the loop at state $s$ is removed.
  The state $s$ has in-degree $0$ and will be seen as initial state of $H(\varphi)$.
  The new state $t$ will be used as sink state.
  Therefore, we add the arcs $(q_1,t)$ and $(t,t)$.
  This results in $R'=(R^{-1}-\{(s,s)\})\cup\{(q_1,t),(t,t)\}$.
  Finally, we add a new variable $e$ that is \textit{true} only in state $t$,
  and a variable $d$ that is \textit{true} in states $\LIT(\psi)\cup\{s\}$.
  For all other variables, $\eta'$ is the same as $\eta$.
  \begin{figure}[h]
  {

  \centering

  \includegraphics[scale=0.8]{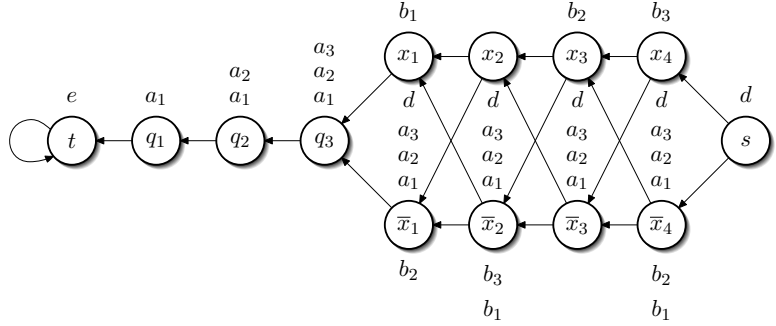}

  }

  \caption{Structure $H(\psi)$ for 
  $\psi=(x_1\vee \neg x_2 \vee \neg x_4) \wedge (\neg x_1\vee x_3 \vee \neg x_4) \wedge (\neg x_2 \vee x_4)$}
  \label{figure:structure for since}
  \end{figure}
  (Figure~\ref{figure:structure for since} shows an example.)

  The formulae $\varphi^m_1, \varphi^m_2, \ldots,\varphi^m_{m+1}$
  are defined inductively as follows.
  $$\varphi_{m+1}^m ~ = ~ d \mbox{~ ~ ~ and ~ ~ ~} 
  \varphi_{i}^m ~ = ~ \left( \left(a_{i} \S b_{i}\right)\,\S\,\varphi_{i+1}^m\right) ~\mbox{for}~ i=1,2,\ldots,m.$$
  The reduction from $\threesat$ to \dt{\{\G,\S\},\emptyset} is 
  the mapping $\psi \mapsto \mcinst{\G(e \S \varphi_1^m),H(\psi),s}$,
  where $\psi$ is a 3CNF-formula with $m$ clauses.
  This reduction can evidently be performed in logarithmic space.
  To prove its correctness,
  we use the following claim.
  Every path $p=(s,l_n,\ldots,l_1,q_m,\ldots,q_1,t,t,\ldots)$ in $H(\psi)$ that begins in state $s$
  corresponds to an assignment $\mathcal{A}_p=\{l_1,\ldots,l_n\}$ to the variables in $\psi$,
  that sets all literals to \textit{true} that appear on $p$.
  For the sake of simplicity,
  we use the notation $\struct{p}{H(\psi)},q_i \vDash \alpha$ for $\struct{p}{H(\psi)},n+m-i+1\vDash\alpha$.

  \begin{Claim}\label{claim:np4}
  Let $H(\psi)$ be constructed from a formula $\psi=C_1\wedge \ldots\wedge C_m$ with $m$ clauses,
  and let $p=(s,l_n,\ldots,l_1,q_m,\ldots,q_1,t,t,\ldots)$ be a path in $H(\psi)$.
  For $j=1,2,\ldots,m$ it holds that
  \begin{center}
  $\struct{p}{H(\psi)},q_j\vDash \varphi_j^m$  ~ if and only if ~
  the assignment $\mathcal{A}_p$ 
  satisfies clauses $C_j,\ldots,C_m$.
  \end{center}
  \end{Claim}

  \begin{Proofofclaim}[\ref{claim:np4}]
  Notice that $\mathcal{A}_p$ satisfies clause $C_j$ if and only if
  $p$ contains a state $w$ with $b_j\in \eta'(w)$.
  We prove the claim by induction.
  Since the variable $d$ holds in all predecessors of $q_m$ in $p$ but not in $q_m$,
  it follows that $\varphi_m^m = (a_m\S b_m) \S d$ holds in $q_m$ iff $a_m\S b_m$ holds in $q_m$.
  Since $b_m\not\in\eta'(q_m)$, it follows that
  $a_m\S b_m$ holds in $q_m$ iff $b_m$ holds in a predecessor of $q_m$
  iff $\mathcal{A}_p$ satisfies $C_m$.
  This completes the base case.
  For the inductive step,
  notice that $\struct{p}{H(\psi)},q_j\vDash \varphi_j^m$ iff $\struct{p}{H(\psi)},q_j\vDash a_j \S b_j$
  and $\struct{p}{H(\psi)},q_{j+1}\vDash\varphi_ {j+1}$.
  By the construction of $H(\psi)$ it follows that  $\struct{p}{H(\psi)},q_j\vDash a_j \S b_j$
  iff $\mathcal{A}_p$ satisfies $C_j$, and the rest follows from the induction hypothesis.
  \end{Proofofclaim}

  Finally, let $\psi$ be a 3CNF formula,
  and let $p=(s,l_n,\ldots,l_1,q_m,\ldots,q_1,t,t,\ldots)$ be a path in $H(\psi)$.
  On the first $n+1$ states of $p$, the variable $d$ holds.
  Therefore, $\varphi_1^m$ and henceforth $e\S\varphi_1^m$ is satisfied in all these states.
  On the $m$ following states $q_m,\ldots,q_1$, neither $d$ nor $e$ holds.
  Notice that $\struct{p}{H(\psi)},q_i\vDash \varphi_i^m$ iff $\struct{p}{H(\psi)},q_i\vDash \varphi_{i-1}^m$ (for $i=2,3,\ldots,m$).
  By Claim~\ref{claim:np4}, $\varphi_1^m$
  and henceforth $e\S\varphi_1^m$ is satisfied in all these states iff $\mathcal{A}_p$ satisfies $\psi$.
  On the remaining states, only the variable $e$ holds.
  Hence, $e\S\varphi_1^m$ is satisfied in all the latter states iff $\mathcal{A}_p$ satisfies $\psi$.
  Concluding, it follows that $\struct{p}{H(\psi)},0\vDash\G(e\S\varphi_1^m)$ iff $\mathcal{A}_p$ satisfies $\psi$.
  Since for every assignment to $\psi$ the structure $H(\psi)$ contains a corresponding path,
  the correctness of the reduction is proven.
  \end{proof}
\fi


The future-operator $\F$ alone is not powerful enough
to make the since-operator $\S$ \NP-hard:
We will show in Theorem~\ref{theorem:MC(S,F;V) NL-c}
that \dt{\{\F,\S\},B} for $\clone B\subseteq\cV$ is \NL-complete.
But with the help of $\neg$ or $\wedge$,
the model-checking problem for $\F$ and $\S$ becomes intractable.

\begin{theorem}\label{theorem:MC(S,F;E|N) NP-h}
 Let $N_{+}$ be a finite set of Boolean functions
      such that $\cN\subseteq\clone{N_{+}}$.
Then \dt{\{\F,\S\},N_{+}}\ is \NP-hard.
 \end{theorem}

\ifextended
  \begin{proof}
  By Lemma~\ref{lemma:subred} it suffices to give a reduction from \threesat\ 
  to $\dt{\{\F,\S\},\{\neg\}}$.
  For a 3CNF formula $\psi$, let $\mcinst{\G(e\S\varphi^m_1),H(\psi),s}$ be
  the instance of $\dt{\{\G,\S\},\emptyset}$ as described in the proof of Theorem~\ref{theorem:MC(S,G;.) NP-h}.
  Using $\G \alpha \equiv \neg\F\neg\alpha$,
  it follows that $\G(e\S\varphi^m_1)\equiv\neg\F(\neg(e\S\varphi^m_1))$,
  where the latter is a $\cN$-formula over $\{\F,\S\}$.
  The correctness of the reduction the same line as the proof of Theorem~\ref{theorem:MC(S,G;.) NP-h}.
  \end{proof}
\fi


\ifextended
  \begin{theorem}\label{theorem:MC(S,X;.) NP-c}
   Let $B$ be a finite set of Boolean functions. Then \dt{\{\X,\S\},B} is \NP-hard.
  \end{theorem}

  \begin{proof}
  To prove \NP-hardness, we give a reduction from \threesat\ to \dt{\{\X,\S\},\emptyset}.
  For a 3CNF formula $\psi$, let $H(\psi)$ be the structure
  as described in the proof of Theorem~\ref{theorem:MC(S,G;.) NP-h}.
  The reduction function maps $\psi$ to $\mcinst{\X^{n+m+1} \varphi_1,H(\psi),s}$.
  The $\X^{n+m+1}$ ``moves'' to state $q_1$ on any path in $H(\psi)$.
  The correctness proof follows the same line as the proof of Theorem~\ref{theorem:MC(S,G;.) NP-h}.
  \end{proof}
\fi


An upper bound better than \PSPACE\ for the intractable cases with the until-operator
or the since-operator remains open. We will now show that one canonical way to prove an \NP\ upper bound fails, 
in showing that these problems do not have the ``short path property'',
which claims that a path in the structure that fulfills the formula
has length polynomial in the length of the structure and the formula.
Hence, it will most likely be nontrivial to obtain a better upper bound.

We will now sketch such families of structures and formulae
using an inductive definition.
 Let $G_1, G_2, \ldots$ be the family of graphs presented 
in Figures~\ref{fig:exponential size graph 1} and \ref{fig:exponential size graph 2}.
Notice that $G_i$ is inserted into $G_{i+1}$ using the obvious lead-in and lead-out arrows.
  \begin{figure}
  \begin{center}
   \begin{minipage}{0.4\textwidth}
     \includegraphics{fig-exp_size_graph.1}
     \caption{The graph $G_1$}
     \label{fig:exponential size graph 1}
    \end{minipage}
    \begin{minipage}{0.55\textwidth}
      \includegraphics{fig-exp_size_graph.3}
      \caption{The graph $G_{i+1}$}
      \label{fig:exponential size graph 2}
    \end{minipage}
    \end{center}
  \end{figure}
  The truth assignments for these graphs are as follows:
  $$
  G_1: \begin{array}{r|l}
  x^1_1 & b_1 \\
  \hline
  x^1_2 & a_1 \\
  \hline
  x^1_3 & a_1,c_1 \\
  \end{array} \ \ \
  G_{i+1}: \begin{array}{r|l}
  \displaystyle x^{i+1}_1 & \bigwedge_{j=1}^{i+1}a_j \\
  \hline
  \displaystyle x^{i+1}_2 & a_{i+1} \\
  \hline
  \displaystyle x^{i+1}_3 & \bigwedge_{j=1}^{i+1}a_j, c_{i+1} \\
  \hline
  x\in G_i & \mathtext{truth assignment from } G_i, b_{i+1}
  \end{array}
  $$
  Now the formulae are defined as follows: 
\begin{gather*}
\varphi_1 = (a_1\U b_1)\U c_1, \mbox{ ~ ~ and ~ ~ }
\varphi_{i+1} = ((a_{i+1}\U\varphi_i)\U b_{i+1})\U c_{i+1}.
\end{gather*}

  The rough idea behind the construction is as follows: 
To satisfy the formula $\varphi_1$ in $G_1,$ the path has to repeat the circle once. 
In the inductive construction, this leads to an exponential number of repetitions. 

\section{The good fragments: tractability results}
\label{sec:nl}

This subsection is concerned with fragments of LTL that have a tractable model-checking problem.
We will provide a complete analysis for these fragments by proving that model checking
for all of them is \NL-complete or even solvable in logarithmic space.
This exhibits a surprisingly large gap in complexity between easy and hard fragments.

The following lemma establishes \NL-hardness for all tractable fragments.
\ifextended
\else
  It is proven in \cite{bms+07}.
\fi

\begin{lemma}
  \label{lemma:MC(F|G|X;.) NL-h}
  Let $B$ be a finite set of Boolean functions.
  Then \mc{\{\F\},B}, \linebreak
  \mc{\{\G\},B}, and \mc{\{\X\},B} are \NL-hard.
\end{lemma}

\ifextended
  \begin{proof}
    First consider \mc{\{\F\},B}.
    We reduce the accessibility problem for digraphs, GAP, to \mc{\{\F\},\emptyset}.
    The reduction is via the following logspace computable function.
    Given an instance $\mcinst{G,a,b}$ of GAP, where $G=(V,E)$ is a digraph and $a,b \in V$,
    map it to the instance $\mcinst{\F y, K(G), a}$ of \mc{\{\F\},\emptyset}
    with $K(G) = (V, E^+, \eta)$, where $E^+$ denotes the reflexive closure of $E$,
    and $\eta$ is given by $\eta(b) = \{y\}$ and $\eta(v) = \emptyset$, for all $v \in V - \{b\}$.
    It is immediately clear that there is a path from $a$ to $b$ in $G$
    if and only if there is a path $p$ in $K(G)$ starting from $a$ such that $\struct{p}{K(G)},0 \vDash \F y$.
  
    \par\medskip\noindent
    For \mc{\{\X\},B}, we use an analogous reduction from GAP to \mc{\{\X\},\emptyset}.
    Given an instance $\mcinst{G,a,b}$ of \GAP, where $G = (V,E)$,
    transform it into the instance $\mcinst{\X^{|V|}y, K(G), a\big}$ of \mc{\{\X\},\emptyset}
    with the Kripke structure $K(G)$ from above.
    Now it is clear that there is a path from $a$ to $b$ in $G$
    if and only if there is a path of length $|V|$ from $a$ to $b$ in the reflexive structure $K(G)$,
    if and only if there is a path $p$ in $K(G)$ starting from $a$ such that $\struct{p}{K(G)},0 \vDash \X^{|V|}y$.
  
    \par\medskip\noindent
    Now consider \mc{\{\G\},B}.
    We reduce the following problem to \mc{\{\G\},\emptyset}.
    Given a directed graph $G=(V,E)$ and a vertex $a \in V$, is there an infinite path
    in $G$ starting at $a$? It is folklore that this is an \NL-hard problem
    (see Lemma~\ref{lem:existence_of_infinite_path_is_NL-hard} in the Appendix).
    Given an instance $\mcinst{G,a}$ of this problem, transform it into
    the instance $\mcinst{\G y, K'(G), a}$ of \mc{\{\G\},\emptyset}, where $K'(G) = (V',E',\eta)$.
    Here
    $V' = V \,\dot\cup\, \{\tilde{v} \mid v \in V,~ v~\text{has no successor in}~V\}$,
    $E' = E \cup \{(v,\tilde{v}),~(\tilde{v},\tilde{v}) \mid \tilde{v} \in V'\}$,
    $\eta(v) = y$ for all $v \in V$, and $\eta(\tilde{v}) = \emptyset$, for all $\tilde{v} \in V'$.
    It is immediately clear that there is an infinite path in $G$ starting at $a$ if and only if
    there is a path $p$ in $K'(G)$ starting from $a$ such that $\struct{p}{K'(G)},0 \vDash \G y$.
  \end{proof}
\fi

It now remains to establish upper complexity bounds. Let $C$ be one of the clones \cN, \cE, \cV, and \cL,
and let $B$ be a finite set of Boolean functions such that $\clone B \subseteq C$.
Whenever we want to establish \NL-membership for some problem \mc{\cdot,B}, it will suffice to assume
that formulae are given over one of the bases $\{\neg,0,1\}$, $\{\wedge,0,1\}$, $\{\vee,0,1\}$, or
$\{\oplus,0,1\}$, respectively. This follows since these clones only contain constants, projections, and multi-ary versions of \NOT, \AND, \OR, and $\oplus$, respectively.

\begin{theorem}\label{theorem:MC(F,G,X;N) NL-c}
  Let $N_{-}$ be a finite set of Boolean functions such that $\clone{N_{-}} \subseteq \cN.$
  Then \mc{\{\F,\G,\X\},N_{-}} is \NL-complete.
\end{theorem}

\begin{proof}
  The lower bound follows from Lemma~\ref{lemma:MC(F|G|X;.) NL-h}.
  For the upper bound, first note that for an LTL formula $\psi$ the following equivalences hold:
  $\F\F \psi \equiv \F \psi$, $\G\G \psi \equiv \G \psi$, $\F\G\F \psi\equiv \G\F \psi$, $\G\F\G \psi \equiv \F\G \psi$,
  $\G\psi \equiv \neg\F\neg\psi$, and $\F\psi \equiv \neg\G\neg\psi$. Furthermore, it is possible to interchange $\X$
  and adjacent $\G$-, $\F$-, or $\neg$-operators without affecting satisfiability. Under these considerations,
  each formula $\varphi \in \L{\{\F,\G,\X\},N_{-}}$ can be transformed without changing satisfiability
  into a normal form
  {\boldmath$
      \varphi' = \X^mP\!\!\sim\!\!y
  $},
  where $P$ is a prefix ranging over the values ``empty string'', \F, \G, \F\G, and \G\F; $m$ is the number
  of occurrences of \X\ in $\varphi$; $\sim$ is either the empty string or $\neg$; and $y$ is a variable
  or a constant. This normal form has two important properties. First, it can be represented in logarithmic
  space using two binary counters $a$ and $b$. The counter $a$ stores $m$,
  and $b$ takes on values $0,\dots,9$ to represent each possible combination of $P$ and $\sim$.
  Note that $a$ takes on values less than $|\varphi|$, and $b$ has a constant range. Hence both
  counters require at most logarithmic space. It is not necessary to store any information about $y$, because
  it can be taken from the representation of $\varphi$.

  Second, $\varphi'$ can be \textit{computed} from $\varphi$ in logarithmic space. The value of $a$ is obtained
  by counting the occurrences of \X\ in $\varphi$, and $b$ is obtained by linearly parsing $\varphi$ with the automaton
  that is given in Figure~\ref{fig:automaton_for_GFneg_prefix}, and which ignores all occurrences of \X.
  \begin{figure}
    \centering
    \includegraphics{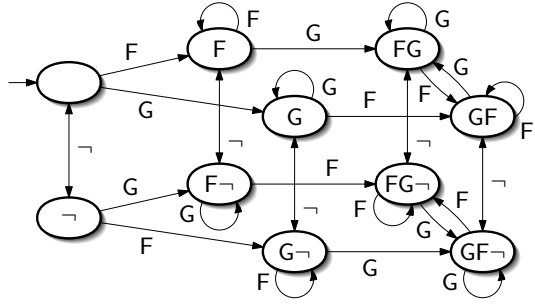}
    \caption{An automaton that computes $P\!\sim$}
    \label{fig:automaton_for_GFneg_prefix}
  \end{figure}

  \ifextended
    The state of this automaton at the end of the passage through $\varphi$ determines the values of $P$ and $\sim$
    in $\varphi$.
    Now let $\varphi$ be an \L{\{\F,\G,\X\},N_{-}}-formula, $K=(W,R,\eta)$ a Kripke structure and $a \in W$.
    If $y$ is constant, the problem is trivial, therefore it remains to consider the case where $y$ is a variable.
    According to the possible values of $P$ and $\sim$ in $\varphi$, there are ten cases to consider.
    We only present the argumentation for those five in which $\sim$ is empty. (For the dual cases,
    kindly replace each occurrence of ``$\in\eta(b)$'' by ``$\not\in\eta(b)$''.)
    In the following list, we assume that $m=0$. As per explanation below, this is not a significant
    restriction.
    \begin{description}
      \item[\boldmath $P$ is empty]
        Then $\mcinst{\varphi,K,a} \in \mc{\{\F,\G,\X\},N_{-}}$ if and only if there is a state $b$ in $K$
        accessible from $a$ via $R$ such that $y\in\eta(b)$.
      \item[\boldmath $P = \F$]
        In this case we have to check whether there is a state $b \in W$ that can be reached from $a$ via
        $R$, and $y\in\eta(b)$.
      \item[\boldmath $P = \G$]
        We define $W'=\{b \in W\mid y\in\eta(b)\}$ and $R'=R\cap W'\times W'$.
        It holds that $\mcinst{\varphi,K,a} \in \mc{\{\F,\G,\X\},N_{-}}$ if and only if there is some
        $b \in W'$ such that $b$ is accessible from $a$ via $R'$ and $b$ belongs to a cycle in $R'$.
      \item[\boldmath $P = \F\G$]
        We can reduce this case to the previous one: $\mcinst{\varphi,K,a} \in \mc{\{\F,\G,\X\},N_{-}}$
        if and only if there is some $b \in W'$ that can be reached from $a$ via $R$, and
        $\mcinst{\G y,K,b}\in\mc{\{\F,\G,\X\},N_{-}}$.
      \item[\boldmath $P = \G\F$]
        We have to check whether there exists some $b \in W$ that can be reached from $a$ via $R$ such that
        $y\in\eta(b)$ and $b$ belongs to a cycle.
    \end{description}
    Since the questions whether there is a path from any vertex to another and whether any vertex belongs to
    a cycle in a directed graph can be answered in \NL, all previously given procedures are \NL-algorithms.
    The restriction $m=0$ is removed by the observation that
    $\mcinst{\X^mP\!\!\sim\!\!y,K,a} \in \mc{\{\F,\G,\X\},N_{-}}$ if and only if there exists some state
    $b$ in $K$ that is accessible from $a$ in $m$ $R$-steps such that
    $\mcinst{P\!\!\sim\!\!y,K,b} \in \mc{\{\F,\G,\X\},N_{-}}$. This reduces the case $m>0$ to $m=0$.
  
    Hence we have found an \NL-algorithm deciding \mc{\{\F,\G,\X\},N_{-}}: Given $\mcinst{\varphi,K,a}$,
    compute $\varphi'$, guess a state $b$ accessible from $a$ in $m$ $R$-steps, apply the procedure
    of one of the above five cases to $\mcinst{\varphi',K,a}$, and accept if the last step was successful.
  \else
    We show in \cite{bms+07}
    how to obtain an \NL\ algorithm using the normal form $\varphi'$.
  \fi
\end{proof}

\begin{theorem}\label{theorem:MC(F,X;V) MC(G,X;E) NL-c}
  \begin{Enum}
   \item
     \label{part:FXV}
     Let $V_{-}$ be a finite set of Boolean functions such that $\clone{V_{-}}\subseteq\cV$.
     Then \mc{\{\F,\X\},V_{-}}\ is \NL-complete.
   \item
     \label{part:GXE}
     Let $E_{-}$ be a finite set of Boolean functions such that $\clone{E_{-}}\subseteq\cE$.\\
     Then \mc{\{\G,\X\},E_{-}}\ is \NL-complete.
  \end{Enum}
\end{theorem}
\begin{proof}
  The lower bounds follow from Lemma~\ref{lemma:MC(F|G|X;.) NL-h}.

  \par\medskip\noindent
  First consider the case $\clone{V_{-}}\subseteq\cV.$  It holds that
  $\F (\psi_1 \vee \dots \vee \psi_n)\equiv \F \psi_1 \vee \dots \vee\F \psi_n$
  as well as $\X\F\varphi \equiv \F\X\varphi$ and $\X(\varphi \vee \psi) \equiv \X\varphi \vee \X\psi$.
  Therefore, every formula $\varphi\in \L{\{\F,\X\},V_{-}}$ can be rewritten as
  \[
      \varphi' = \F\X^{i_1} y_1 \vee \dots \vee \F\X^{i_n} y_n \vee \X^{i_{n+1}}y_{n+1} \vee \dots \vee \X^{i_m}y_m,
  \]
  where $y_1,\dots,y_m$ are variables or constants (note that this representation of $\varphi$ can be constructed
  in \LS). Now let $\mcinst{\varphi,K,a}$ be an instance of \mc{\{\F,\X\},V_{-}}, where $K = (W,R,\eta)$,
  and let $\varphi$ be of the above
  form. Thus, $\mcinst{\varphi,K,a} \in \mc{\{\F,\X\},V_{-}}$ if and only if for some $j \in \{n+1,\dots,m\}$, there is a
  state $b \in W$ such that $y_j \in \eta(b)$ and $b$ is accessible from $a$ in exactly $i_j$ $R$-steps or
  if, for some $j \in \{1,\dots,n\}$, there is a state $b \in W$ such that $y_j \in \eta(b)$ and $b$ is accessible
  from $a$ in at least $i_j$ $R$-steps. This can be tested in \NL.

  \par\medskip\noindent
  As for the case $\clone{E_{-}}\subseteq\cE,$ we take advantage of the duality of $\F$ and $\G$, and $\wedge$ and $\vee$,
  respectively. Analogous considerations as above lead to the logspace computable normal form
  \[
      \varphi' = \G\X^{i_1} y_1 \wedge \dots \wedge \G\X^{i_n} y_n \wedge
                 \X^{i_{n+1}}y_{n+1} \wedge \dots \wedge \X^{i_m}y_m.
  \]
  Let $I = \max\{i_1,\dots,i_m\}$. For each $j = 1,\dots,m$, we define $W^j=\{b \in W\mid y_j\in\eta(b)\}$
  and $R^j = R \cap W^j\times W^j$. Furthermore, let $W'$ be the union of $W^j$ for $j = 1,\dots,n$ (!), and let
  $R' = R \cap W'\times W'$. Now $\mcinst{\varphi,K,a} \in \mc{\{\G,\X\},E_{-}}$
  if and only if there is some state $b \in W'$ satisfying the following conditions.
  \begin{itemize}
    \item
      There is an $R$-path $p$ of length at least $I$ from $a$ to $b$, where the first $I+1$ states
      on $p$ are $c_0 = a$, $c_1$, \dots, $c_{I}$.
    \item
      The state $b'$ lies on a cycle in $W'$.
    \item
      For each $j = 1,\dots,n$, each state of $p$ from $c_{i_j}$ to $c_{I}$ is from $W^j$.
    \item
      For each $j = n+1,\dots,m$, the state $c_{i_j}$ is from $W^j$.
  \end{itemize}
  These conditions can be tested in \NL\ as follows. Successively guess $c_1,\dots,c_{I}$ and verify
  their membership in the appropriate sets $W^j$. Then guess $b$, verify whether $b \in W'$, whether
  $b$ lies on some $R'$-cycle, and whether there is an $R'$-path from $c_{I}$ to $b$.
\end{proof}

In the proof of Theorem~\ref{theorem:MC(F,X;V) MC(G,X;E) NL-c},
we have exploited the duality of \F\ and \G, and $\vee$ and $\wedge$, respectively.
Furthermore, the proof relied on the fact that \F\ and $\vee$ (and \G\ and $\wedge$) are interchangeable.
This is not the case for \F\ and $\wedge$, or \G\ and $\vee$, respectively.
Hence it is not surprising that $\mc{\{\F\},\set{\wedge}}$ is \NP-hard (Corollary~\ref{cor:MC(F;E) NP-h}). However, the \NL-membership of $\mc{\{\F,\G\},\set\vee}$ \textit{is} surprising.
Before we formulate this result, we try to provide an intuition for the tractability of this problem.
The main reason is that an inductive view on $\L{\{\F,\G\},\set\vee}$-formulae allows us
to subsequently guess parts of a satisfying path without keeping the previously guessed parts in memory.
This is possible because each $\L{\{\F,\G\},\set\vee}$-formula $\varphi$ can be rewritten as
\begin{equation}
  \label{eq:normal_form_FGV}
  \varphi = y_1 \vee \dots \vee y_n \vee \F z_1 \vee \dots \vee \F z_m \vee \G\psi_1 \vee \dots \vee \G\psi_\ell
            \vee \F\G\psi_{\ell+1} \vee \dots \vee \F\G\psi_k,
\end{equation}
where the $y_i,z_i$ are variables (or constants),
and each $\psi_i$ is an $\L{\{\F,\G\},\set\vee}$-formula of the same form with a strictly smaller nesting depth of \G-operators.
Now, $\varphi$ is \textit{true} at the begin of some path $p$ iff one of its disjuncts is \textit{true} there.
In case none of the $y_i$ or $\F z_i$ is \textit{true}, we must guess one of the $\G\psi_i$ (or $\F\G\psi_j$)
and check whether $\psi_i$ (or $\psi_j$) is \textit{true} on the entire path $p$ (or on $p$ minus some finite number
of initial states).
Now $\psi_i$ is again of the above form.
So we must either find an infinite path
on which $y_1 \vee \dots \vee y_n \vee \F z_1 \vee \dots \vee \F z_m$ is \textit{true} everywhere
(a cycle containing at least $|N|$ states satisfying some $y_i$ or $z_i$ suffices,
where $N$ is the set of states of the Kripke structure),
or we must find a \textit{finite} path satisfying the same conditions
and followed by an infinite path satisfying one of the $\G\psi_i$ (or $\F\G\psi_j$) at its initial point.
Hence we can recursively solve a problem of the same kind with reduced problem size.
Note that it is neither necessary to explicitly compute the normal form for $\varphi$ or one of the $\psi_i$,
nor need previously visited states be stored in memory. 


\begin{theorem}
  \label{theorem:MC(F,G;V) NL-c}
  Let $V_{-}$ be a finite set of Boolean functions such that $\clone{V_{-}}\subseteq\cV.$
  Then \mc{\{\G\},V_{-}} and \mc{\{\F,\G\},V_{-}} are \NL-complete.
\end{theorem}

\begin{proof}
  The lower bound follows from Lemma~\ref{lemma:MC(F|G|X;.) NL-h}.
  It remains to show \NL-membership of \mc{\{\F,\G\},V_{-}}.
  For this purpose, we devise the recursive algorithm \FGValgo\ as given in Table \ref{tab:alg_ltl-F,G,V-mc}.
  Note that we have deliberately left out constants.
  This is no restriction, since we have observed in Lemma~\ref{lemma:const} that each constant can be regarded as
  a variable that is set to \textit{true} or \textit{false} throughout the whole Kripke structure.

  \begin{table}
    \centering
    \renewcommand{\arraystretch}{1.1}
    \begin{minipage}{.7\textwidth}
      \textbf{Algorithm {\boldmath \FGValgo}}

      \par\medskip
      \begin{tabular}{@{}l@{\qquad}l@{}}
        \textbf{Input}  & $\varphi \in \L{\{\F,\G\},V_{-}}$                                           \\
                        & Kripke structure $K = (W,R,\eta)$                                         \\
                        & $a \in W$                                                                 \\
                        & additional parameter $\textit{mode} \in \{\texttt{now},\texttt{always}\}$ \\[2pt]
        \textbf{Output} & \textbf{accept} or \textbf{reject}
      \end{tabular}

      \par\medskip
      \begin{algorithmic}[1]
        \STATE{%
          $c \leftarrow 0$; \quad
          $\psi \leftarrow \varphi$; \quad
          $b \leftarrow a$; \quad
          $\textit{Ffound} \leftarrow \texttt{false}$%
        }
        \WHILE{$c \le |W|$}
          \IF{$\psi = \alpha_0 \vee \alpha_1$\quad (for some $\alpha_0,\alpha_1$)\quad}
            \STATE{guess $i \in \{0,1\}$}
            \STATE{$\psi \leftarrow \alpha_i$}
          \rule[-10pt]{0pt}{5pt}
          \ELSIF{$\psi = \F\alpha$\quad (for some $\alpha$)\quad}
            \STATE{$\textit{Ffound} \leftarrow \texttt{true}$}
            \STATE{$\psi \leftarrow \alpha$}
          \rule[-10pt]{0pt}{5pt}
          \ELSE[$\psi$ is some $\G\alpha$ or a variable]
            \IF[process encountered \F]{\textit{Ffound}}
              \STATE{guess $n$ with $0 \le n \le |W|$}
              \FOR[if $n=0$, ignore this loop]{$i = 1,2,\dots,n$}
                \STATE{$b \leftarrow \text{guess some $R$-successor of $b$}$}
              \ENDFOR
            \ENDIF
            \rule[-10pt]{0pt}{5pt}
            \IF{$\psi = \G\alpha$\quad (for some $\alpha$)\quad}
              \STATE{\textbf{call} $\FGValgo(\alpha,K,b,\texttt{always})$}
            \rule[-10pt]{0pt}{5pt}
            \ELSE[$\psi$ is a variable]
              \IF{$\psi \notin \eta(b)$}
                \STATE{\textbf{reject}}
              \ENDIF
              \IF{$\textit{mode} = \texttt{always}$}
                \STATE{$c \leftarrow c+1$}
                \STATE{$b \leftarrow \text{guess some $R$-successor of $b$}$}
                \STATE{$\textit{Ffound} \leftarrow \texttt{false}$}
                \STATE{$\psi \leftarrow \varphi$}
              \ELSE
                \STATE{\textbf{accept}}
              \ENDIF
            \rule[-10pt]{0pt}{5pt}
            \ENDIF
          \ENDIF
        \ENDWHILE
        \STATE{\textbf{accept}}
      \end{algorithmic}
    \end{minipage}
    \renewcommand{\arraystretch}{1.5}
    \caption{The algorithm \FGValgo}
    \vspace{-.5\baselineskip}
    \label{tab:alg_ltl-F,G,V-mc}
  \end{table}

  The parameter \textit{mode} indicates the current ``mode'' of the computation. The idea is as follows.
  In order to determine whether $\varphi$ is satisfiable at the \textit{initial} point
  of some structure starting at $a$ in $K$, the algorithm has to be in mode \texttt{now}.
  This, hence, is the default setting for the first call of \FGValgo. As soon as
  the algorithm chooses to process a \G-subformula $\G\alpha$ of $\varphi$, it has to determine
  whether $\alpha$ is satisfiable at \textit{every} point in some structure starting
  at the currently visited state in $K$. It therefore changes into \texttt{always} mode
  and calls itself recursively with the first parameter set to $\alpha$, see Line 17.

  Hence, given an instance $\mcinst{\varphi,K,a}$ of the problem \mc{\{\F,\G\},V_{-}}, we have to invoke
  $\FGValgo(\varphi,K,a,\texttt{now})$ in order to determine whether there is
  a satisfying path for $\varphi$ in $K$ starting at $a$. It is easy to see that
  this call always terminates: First, whenever the algorithm calls itself recursively,
  the first argument of the new call is a strict subformula of the original first argument.
  Therefore there can be at most $|\varphi|$ recursive calls. Second, within each call,
  each passage through the \textit{while} loop (Lines 2--32) either decreases $\psi$
  or increases $c$. Hence, there can be at most $|\varphi|\cdot(|W|+1)$ passages through
  the \textit{while} loop until the algorithm accepts or rejects.

  \FGValgo\ is an \NL\ algorithm: The values of all parameters and
  programme variables are either subformulae of the
  original formula $\varphi$, states of the given Kripke structure $K$, counters of range $0,\dots,|W|+1$,
  or Booleans. They can all be represented using $\lceil\log|\varphi|\rceil$,
  $\lceil\log(|W|+1)\rceil$, or constantly many bits. Furthermore, since the algorithm uses
  no \textit{return} command, the recursive calls may re-use the space provided for all
  parameters and programme variables, and no return addresses need be stored.

  It remains to show the correctness of \FGValgo, which we will do in two steps.
  in \texttt{always} mode, which will be shown by induction on the nesting depth
  of the \G-operator in $\varphi$. We denote this value by $\mu_\G(\varphi)$.
  Claim~\ref{claim:FGV2} will then ensure the correct behaviour in \texttt{now} mode.
  \ifextended
  \else
    Both claims are proven in \cite{bms+07}.
  \fi

  \begin{Claim}
    \label{claim:FGV1}
    For each $\varphi \in \L{\{\F,\G\},\cV}$, each $K = (W,R,\eta)$, and each $a \in W$:
    
      $\mcinst{\G\varphi,K,a} \in \mc{\{\F,\G\},V_{-}}$
      ~~
      $\Leftrightarrow$
      ~~
      \text{there is an accepting run of~}\\ \mbox{}\hfill
      $\FGValgo(\varphi,K,a,\texttt{always})$.
    
  \end{Claim}%
  \ifextended

    \begin{Proofofclaim}[\ref{claim:FGV1}]
      \textit{For the base case of the induction,}~ let $\mu_\G(\varphi) = 0$. Because of the equivalences
      $\F(\psi_1 \vee \psi_2) \equiv \F\psi_1 \vee \F\psi_2$ and $\F\F\psi \equiv \F\psi$,
      we may assume w.l.o.g.\ that any occurrence of the \F-operator is in front of some variable in
      $\varphi$. If we think of $\varphi$ as a tree, this means that \F-operators can only occur
      in direct predecessors of leaves.
      Note that the algorithm computes this normal form implicitly: Whenever it guesses a path from the root ($\varphi$)
      to some leaf (a variable) in the tree
      and encounters an \F-operator in Line 6, the flag \textit{Ffound} is set. Only after processing all
      $\vee$-operators on the remaining part of the path, the \F-operator is processed in Lines 10--15.
      Now let $\VAR_1(\varphi)$ be all variables that occur in the scope of an \F-operator in $\varphi$,
      and let $\VAR_0(\varphi)$ be all other variables in $\varphi$.
  
      \par\smallskip\noindent
      \textit{For the {\boldmath``$\Rightarrow$''} direction,}~ suppose $\mcinst{\G\varphi,K,a} \in \mc{\{\F,\G\},V_{-}}$.
      Then there exists a path $p$ in $K$ such that $p_0 = a$, and for all $i \ge 0$,
      $\struct{p}{K},i \vDash \varphi$. This means that, for each $i$, either there exists some $x_i \in \VAR_0(\varphi)$
      such that $\struct{p}{K},i \vDash x_i$, or there is some $x_i \in \VAR_1(\varphi)$ such that
      $\struct{p}{K},i \vDash \F x_i$. Now it can be seen that there is a non-rejecting sequence of runs through the \textit{while} loop
      in Lines 2--32 after which $c$ has value $|W|+1$, which then leads to the \textit{accept} in Line 33:
  
      Consider
      the begin of an arbitrary single run through the \textit{while} loop in Line 2. Let $p_i$ be the
      current value of $b$. If $x_i \in \VAR_0(\varphi)$, then the algorithm can ``guess its way through
      the tree of $\varphi$'' in Lines 3--5 and finally reaches Line 19 with $\psi = x_i$.
      It does not reject in Line 20, increases $c$ in Line 23, guesses $p_{i+1}$ in
      Line 24, and resets \textit{Ffound} and $\psi$ appropriately in Lines 25, 26. Otherwise,
      if $x_i \in \VAR_1(\varphi)$, then there is some $n \ge 0$ such that $p_{i+n}$ satisfies
      $x_i$. It is safe to assume that $n \le |W|$ because otherwise the path from $p_i$ to $p_{i+n}$
      would describe a cycle within $K$ which could be replaced by a shorter, more direct, path
      without affecting satisfiability of the relevant subformulae in the states $p_0,\dots,p_i$.
      Now the algorithm can proceed as in the previous case, but, in addition, it has to guess the
      correct value of $n$ and the sequence $p_{i+1},\dots,p_{i+n}$ in Lines 10--15.
  
      \par\smallskip\noindent
      \textit{For the {\boldmath``$\Leftarrow$''} direction,}~ let there be an accepting run of
      $\FGValgo(\varphi,K,a,\texttt{always})$. Since the algorithm is in \texttt{always}
      mode, and $\varphi$ is \G-free, the acceptance can only take place in Line 33, without a recursive
      call in Line 17.
      Hence the counter $c$ reaches value $|W|+1$
      in the \textit{while} loop in Lines 2--32.
  
      Let $p=p_0,p_1,\dots,p_m$ be the sequence of states guessed in this run in Lines 13 and 24,
      where $p_0 = a$. Furthermore, let $i_0,\dots,i_{|W|+1}$ be an index sequence that determines a subsequence
      of $p$ such that
      \begin{itemize}
        \item
          $0 = i_0 < i_1 < \dots < i_{|W|+1} = m$,\quad and
        \item
          for each $j > 0$, $p_{i_j}$ is the value assigned to $b$ in Line 24
          after having set $c$ to value $j$ in Line 23.
      \end{itemize}
      Now it is clear that for all $j = 0,\dots,|W|$, there must be a variable $x_j$ such that
      $x_j \in \eta(p_{i_{j+1}-1})$. If $x_j \in \VAR_0(\varphi)$, then $p_{i_{j+1}} = p_{i_j}+1$,
      and each structure $p'$ extending $p$ beyond $p_m$ satisfies $x_j$ (and hence $\varphi$)
      at $p_{i_j}$. Otherwise $x_j \in \VAR_1(\varphi)$, and the accepting run of the algorithm has guessed
      the states $p_{i_j},\dots,p_{i_{j+1}-1}$ in Line 13. In this case, each structure $p'$ extending $p$
      beyond $p_m$ satisfies $\F x_j$ (and hence $\varphi$) at $p_{i_j},\dots,p_{i_{j+1}-1}$.
      From these two cases, we conclude that each such $p'$ satisfies $\varphi$ in all states $p_0,\dots,p_m$.
  
      We now restrict attention to the states $p_{i_1-1},\dots,p_{i_{|W|+1}-1}$. Among these $|W|+1$ states,
      some of the $|W|$ states of $K$ has to occur twice. Assume $p_{i_j-1}$ and $p_{i_k-1}$ represent the
      same state from $K$, where $j < k$. Then we can create an (infinite) structure $p''$ from $p$ that consists
      of states $p_0,\dots,p_{i_k-1}$, followed by an infinite repetition of the sequence
      $p_{i_j},\dots,p_{i_k-1}$. It is now obvious that $p''$ satisfies $\varphi$ in every state, hence
      $p'',0 \vDash \varphi$, that is, $\mcinst{\G\varphi,K,a} \in \mc{\{\F,\G\},V_{-}}$.
  
      \par\medskip\noindent
      \textit{For the induction step,}~ let $\mu_\G(\varphi) > 0$. For the same reasons as above,
      we can assume that any \F-operator only occurs in front of variables or in front of some \G-operator
      in $\varphi$. This ``normal form'' is taken care of by setting \textit{Ffound} to
      \texttt{true} when \F\ is found (Line 7) and processing this occurrence of \F\ only when a variable
      or some \G-operator is found (Lines 10--15).
  
      \par\smallskip\noindent
      \textit{For the {\boldmath``$\Rightarrow$''} direction,}~ suppose $\mcinst{\G\varphi,K,a} \in \mc{\{\F,\G\},V_{-}}$.
      Then there exists a path $p$ in $K$ such that $p_0 = a$, and for all $i \ge 0$,
      $\struct{p}{K} \vDash \varphi$. We describe an accepting run of $\FGValgo(\varphi,K,a,\texttt{always})$.
      Consider a single passage through the \textit{while} loop with the following configuration. The programme
      counter has value $2$, $c$ has value at most $|W|$, $b$ has value $p_i$, and $\psi$ has value $\varphi$.
      Since $\struct{p}{K} \vDash \varphi$, there are four possible cases. The argumentation for the first two of them
      is the same as in the base case.
  
      \newlength{\links}
      \settowidth{\links}{\textsl{Case 4.}~~}
      \begin{list}{
        blub
      }{
        \leftmargin\links
        \labelwidth\links
      }
        \item[\textsl{Case 1.}]
          $\struct{p}{K},i \vDash x$, for some $x \in \VAR_0(\varphi)$.
  
        \item[\textsl{Case 2.}]
          $\struct{p}{K},i \vDash \F x$, for some $x \in \VAR_1(\varphi)$.
  
        \item[\textsl{Case 3.}]
          $\struct{p}{K},i \vDash \G\alpha$, for some maximal \G-subformula $\G\alpha$ of $\varphi$
          that is \textit{not} in the scope of some \F-operator.
  
          This means that $\alpha$ is \textit{true} everywhere on the path $p_i,p_{i+1},p_{i+2},\dots$. Hence,
          due to the induction hypothesis, $\FGValgo(\alpha,K,b_i,\texttt{always})$ has
          an accepting run. By appropriate guesses in Line 4, the current call of the algorithm can reach that
          accepting recursive call in Line 17.
  
        \item[\textsl{Case 4.}]
          $\struct{p}{K},i \vDash \G\alpha$, for some maximal \G-subformula $\G\alpha$ of $\varphi$
          that \textit{is} in the scope of some \F-operator.
  
          By combining the arguments of Cases 3 and 2, we can find an accepting run for this case.
      \end{list}
  
      \par\noindent
      If only Cases 1 or 2 occur more than $|W|$ times in a sequence, then $c$ will finally take on
      value $|W|+1$, and this call will accept in Line 31. Otherwise, whenever one of Cases 3 and 4
      occurs, than the acceptance of the new call---and hence of the current call---is due
      to the induction hypothesis.
  
      \par\smallskip\noindent
      \textit{For the {\boldmath``$\Leftarrow$''} direction,}~ let there be an accepting run of
      $\FGValgo(\varphi,K,a,\texttt{always})$. Since the algorithm is in \texttt{always}
      mode, the acceptance can only take place in Line 33 or in the recursive call in Line 17. If the run
      accepts in Line 33, the same arguments as in the base case apply. If the acceptance is via the
      recursive call, then let $p = p_0,\dots,p_m$ be the sequence of states guessed such that $p_0 = a$,
      and $p_m$ is the value of $b$ when the recursive call with $\G\alpha$ takes place.
      Due to the induction hypothesis, $\mcinst{\G\alpha,K,b_m} \in \mc{\{\F,\G\},V_{-}}$ and, hence,
      there is an infinite structure $p'$ extending $p$ beyond $p_m$ such that $\struct{(p')}{K},m \vDash \G\varphi$.
      Furthermore, we can use the same argumentation as in the base case to show that,
      for each $i \leq m$, $\struct{(p')}{K},i \vDash \varphi$. Therefore, $\struct{(p')}{K},0 \vDash \G\varphi$,
      which proves $\mcinst{\G\varphi,K,a} \in \mc{\{\F,\G\},V_{-}}$.
    \end{Proofofclaim}
  \fi

  \begin{Claim}
    \label{claim:FGV2}
    For each $\varphi \in \L{\{\F,\G\},V_{-}}$, each $K = (W,R,\eta)$, and each $a \in W$:
    \[
      \mcinst{\varphi,K,a} \in \mc{\{\F,\G\},V_{-}}
      ~~
      \Leftrightarrow
      ~~
      \text{there is an accepting run of~}
      \FGValgo(\varphi,K,a,\texttt{now})
    \]
  \end{Claim}%

  \ifextended
    \begin{Proofofclaim}[\ref{claim:FGV2}]
      \textit{For the {\boldmath``$\Rightarrow$''} direction,}~ suppose $\mcinst{\varphi,K,a} \in \mc{\{\F,\G\},V_{-}}$.
      Then there exists a path $p$ in $K$ such that $p_0 = a$ and $\struct{p}{K},0 \vDash \varphi$.
      We describe an accepting run of $\FGValgo(\varphi,K,a,\texttt{now})$.
      Consider the first passage through the \textit{while} loop with the following configuration. The programme
      counter has value $2$, $c$ has value $0$ (this value does not change in \textit{now} mode),
      $b$ has value $a$, and $\psi$ has value $\varphi$.
      Since $\struct{p}{K},0 \vDash \varphi$, there are four possible cases. The argumentation for them is very similar to
      that in the proof of Claim \ref{claim:FGV1}.
    
      \settowidth{\links}{\textsl{Case 4.}~~}
      \begin{list}{
        blub
      }{
        \leftmargin\links
        \labelwidth\links
      }
        \item[\textsl{Case 1.}]
          $\struct{p}{K},0 \vDash x$, for some $x \in \VAR_0(\varphi)$.
    
          As in the proof of Claim \ref{claim:FGV1}, the algorithm can guess the appropriate disjuncts
          in Lines 3--5, does not reject in Line 20 and accepts (it is in \textit{now} mode!)
          in Line 28.
    
        \item[\textsl{Case 2.}]
          $\struct{p}{K},0 \vDash \F x$, for some $x \in \VAR_1(\varphi)$.
    
          As in the proof of Claim \ref{claim:FGV1}, there exists some $n$ with $0 \le n \le |W|$ such that $b_n$ satisfies
          $x_i$. The algorithm can proceed as in the previous case, but, in addition, it has to guess the
          correct value of $n$ and the sequence $p_1,\dots,p_n$ in Lines 10--15.
    
        \item[\textsl{Case 3.}]
          $\struct{p}{K},0 \vDash \G\alpha$, for some maximal \G-subformula $\G\alpha$ of $\varphi$
          that is \textit{not} in the scope of some \F-operator.
    
          This means that $\alpha$ is \textit{true} everywhere on the path $p$. Hence,
          due to the induction hypothesis, $\FGValgo(\alpha,K,b_i,\texttt{now})$ has
          an accepting run. By appropriate guesses in Line 4, the current call of the algorithm can reach that
          accepting recursive call in Line 17.
    
        \item[\textsl{Case 4.}]
          $\struct{p}{K},0 \vDash \G\alpha$, for some maximal \G-subformula $\G\alpha$ of $\varphi$
          that \textit{is} in the scope of some \F-operator.
    
          By combining the arguments of Cases 3 and 2, we can find an accepting run for this case.
      \end{list}
    
      \par\smallskip\noindent
      \textit{For the {\boldmath``$\Leftarrow$''} direction,}~ suppose there exists an accepting run of
      $\FGValgo(\varphi,K,a,\texttt{now})$. Since the algorithm is in \texttt{now}
      mode, the acceptance can only take place in Line 28 or in the recursive call in Line 17. If the run
      accepts in Line 28, then there is some variable $x$ such that either
      $x \in \VAR_0(\varphi)$ and $x \in \eta(a)$, or $x \in \VAR_1(\varphi)$ and the run guesses a path
      $p_0,\dots,p_m$ with $p_0 = a$ and $x \in \eta(p_m)$. In both cases, each structure $p'$ extending
      the sequence of states guessed so far, satisfies $\varphi$ at $a$.
      On the other hand, if the run accepts in the recursive call, we can argue as in the proof
      of Claim \ref{claim:FGV1}.
    \end{Proofofclaim}
  \fi
\end{proof}

Unfortunately, the above argumentation fails for \mc{\{\G,\X\},V} because of the following considerations.
The \NL-algorithm in the previous proof relies on the fact
that a satisfying path for $\G\psi$, where $\psi$ is of the form \eqref{eq:normal_form_FGV},
can be divided into a ``short'' initial part satisfying the disjunction of the atoms,
and the remaining end path satisfying one of the $\G\psi_i$ at its initial state.
When guessing the initial part, it suffices to separately guess each state and consult $\eta$.

If \X\ were in our language, the disjuncts would be of the form $\X^{k_i}y_i$ and $\X^{\ell_i}\G\psi_i$.
Not only would this make the guessing of the initial part more intricate. It would also require
memory for processing each of the previously satisfied disjuncts $\X^{k_i}y_i$. An adequate modification of
\FGValgo\ would require more than logarithmic space. We have shown \NP-hardness for \mc{\{\G,\X\},V}
in Theorem~\ref{theorem:MC(G,X;V) NP-h}.

\begin{theorem}\label{theorem:MC(X;L) NL-c}
  Let $L_{-}$ be a finite set of Boolean functions such that $\clone{L_{-}}\subseteq\cL.$
  Then \mc{\{\X\},L_{-}} is \NL-complete.
\end{theorem}

  \begin{table}[t]
    \centering
    \renewcommand{\arraystretch}{1.1}
    \textbf{Algorithm {\boldmath \XLalgo}}

    \par\medskip
    \parbox{.49\textwidth}{
      \textbf{Input}                                                                 \\
      \hspace*{7mm}$\varphi' = \X^{i_1}p_1 \oplus \dots \oplus \X^{i_\ell}p_\ell$    \\
      \hspace*{7mm}Kripke structure $K = (W,R,\eta)$                                 \\
      \hspace*{7mm}$a \in W$                                                         \\[4pt]
      \textbf{Output}                                                                \\
      \hspace*{7mm}\textbf{accept} or \textbf{reject}
    }
    \hfill
    \parbox{.49\textwidth}{
      \begin{algorithmic}[1]
        \STATE{%
          $\textit{parity} \leftarrow 0$;~
          $b \leftarrow a$;~
          $k \leftarrow 0$%
        }
        \WHILE{$k \le m$}
          \FOR{$j = 1,\dots,\ell$}
            \IF{$i_j = k$ and $p_j \in \eta(b)$}
              \STATE{$\textit{parity} \leftarrow 1-\textit{parity}$}
            \ENDIF
          \ENDFOR
          \STATE{$k \leftarrow k+1$}
          \STATE{$b \leftarrow \text{guess some $R$-successor of $b$}$}
        \ENDWHILE
        \STATE{\textbf{return} $\textit{parity}$}
      \end{algorithmic}
    }
    \caption{The algorithm \XLalgo}
    \label{tab:XLalgo}
    \renewcommand{\arraystretch}{1.5}
  \end{table}

\begin{proof}
  The lower bound follows from Lemma~\ref{lemma:MC(F|G|X;.) NL-h}.

  For the upper bound, let $\varphi \in \L{\{\X\},L_{-}}$ be a formula, $K = (W,R,\eta)$ a Kripke structure,
  and $a \in W$ a state. Let $m$ denote the maximal nesting depth of \X-operators
  in $\varphi$. Since for any $k$-ary Boolean operator $f$ from $L_{-}$, the formula
  $\X f(\psi_1,\dots,\psi_k)$ is equivalent to $f(\X\psi_1,\dots\X\psi_k)$,
  $\varphi$ is equivalent to a formula $\varphi' \in \L{\{\X\},L_{-}}$ of the form
  {\boldmath$
      \varphi' = \X^{i_1}p_1 \oplus \dots \oplus \X^{i_\ell}p_\ell\,
  $},
  where $0 \le i_j \le m$ for each $j = 1,\dots,\ell$. It is not necessary to compute $\varphi'$
  all at once, because it will be sufficient to calculate $i_j$ each time the variable $p_j$ is encountered
  in the algorithm \XLalgo\ given in Table~\ref{tab:XLalgo}.

  It is easy to see that \XLalgo\ returns $1$ if and only if $\varphi$ is satisfiable.
  From the used variables, it is clear that \XLalgo\ runs in nondeterministic logarithmic space.~~~~
\end{proof}

\par\medskip
In the fragment with \S\ as the only temporal operator,
\S\ is without effect, since we can never leave the initial state.
Hence, any formula $\alpha\S\beta$ is satisfied at the initial state of any structure $K$
if and only if $\beta$ is.
This leads to a straightforward logspace reduction from \mc{\{\S\},\cBF} to \mc{\emptyset,\cBF}:
Given a formula $\varphi \in \L{\{\S\},\cBF}$,
successively replace every subformula $\alpha\S\beta$ by $\beta$ until all occurrences of \S\ are eliminated.
The resulting formula $\varphi'$ is initially satisfied in any structure $K$ iff $\varphi$ is.

Now \mc{\emptyset,\cBF} is the Formula Value Problem,
which has been shown to be solvable in logarithmic space in \cite{lynch77}.
Thus we obtain the following result.

\begin{theorem}\label{theorem:MC(S;BF) in L}\label{theorem:MC(S;.) in L}
  Let $B$ be a finite set of Boolean functions. Then $\mc{\{\S\},B} \in \LS$.
\end{theorem}

In our classification of complexity, which is based on logspace reductions $\redlogm$,
a further analysis of \S-fragments is not possible.
However, a more detailed picture emerges if stricter reductions are considered,
see \cite[Chapter 2]{sch07}.

\begin{theorem}\label{theorem:MC(S,F;V) NL-c}
  Let $V_{-}$ be a finite set of Boolean functions such that $\clone{V_{-}}\subseteq\cV.$
  Then \mc{\{\S,\F\},V_{-}} is \NL-complete.
\end{theorem}

\begin{proof}
  The lower bound follows from Lemma~\ref{lemma:MC(F|G|X;.) NL-h}.
  For the upper bound, we will show that \mc{\{\S,\F\},V_{-}} can be reduced to \mc{\{\F\},V_{-}}
  by disposing of the \S-operator as follows.
  Consider an arbitrary Kripke structure $K$ and a path $p$ therein.
  Then the following equivalences hold.

  \noindent
  \parbox{.44\textwidth}{
    \begin{alignat}{3}
      \struct{p}{K},0 & \vDash \alpha\S\beta       & & \text{~~iff~~} & \struct{p}{K},0 & \vDash \beta                 \label{eq:SFV_S}  \\
      \struct{p}{K},0 & \vDash \F(\alpha\S\beta)   & & \text{~~iff~~} & \struct{p}{K},0 & \vDash \F\beta               \label{eq:SFV_SF}
    \end{alignat}
  }
  \hfill
  \parbox{.54\textwidth}{
    \begin{alignat}{3}
      \struct{p}{K},0 & \vDash \F(\alpha\vee\beta) & & \text{~~iff~~} & \struct{p}{K},0 & \vDash \F\alpha \vee \F\beta \label{eq:SFV_Fv} \\
      \struct{p}{K},0 & \vDash \F\F\alpha          & & \text{~~iff~~} & \struct{p}{K},0 & \vDash \F\alpha              \label{eq:SFV_FF}
    \end{alignat}
  }

  Statements \eqref{eq:SFV_Fv} and \eqref{eq:SFV_FF} are standard properties
  and follow directly from the definition of satisfaction for \F\ and $\vee$.
  Statement \eqref{eq:SFV_S} is simply due to the fact that there is no state in the past of $p_0$.
  As for \eqref{eq:SFV_SF}, we consider both directions separately.
  Assume that $\struct{p}{K},0 \vDash \F(\alpha\S\beta)$.
  Then there is some $i \ge 0$ such that $\struct{p}{K},i \vDash \alpha\S\beta$.
  This implies that there is some $j$ with $0 \le j \le i$ and $\struct{p}{K},j \vDash \beta$.
  Hence, $\struct{p}{K},0 \vDash \F\beta$.
  For the other direction, let $\struct{p}{K},0 \vDash \F\beta$.
  Then there is some $i \ge 0$ such that $\struct{p}{K},i \vDash \beta$.
  This implies $\struct{p}{K},i \vDash \alpha\S\beta$.
  Hence, $\struct{p}{K},0 \vDash \F(\alpha\S\beta)$.

  Now consider an arbitrary formula $\varphi \in \L{\{\S,\F\},V_{-}}$.
  Let $\varphi'$ be the formula obtained from $\varphi$
  by successively replacing the outermost \S-subformula $\alpha\S\beta$ by $\beta$
  until all occurrences of \S\ are eliminated.
  This procedure can be performed in logarithmic space,
  and the result $\varphi'$ is in $\L{\{\F\},V_{-}}$.
  Due to \eqref{eq:SFV_S}--\eqref{eq:SFV_FF}, for any path $p$ in any Kripke structure $K$, it holds that
  $
    \struct{p}{K},0 \vDash \varphi \text{~if and only if~} \struct{p}{K},0 \vDash \varphi'.
  $
  Hence, the mapping $\varphi \mapsto \varphi'$ is a logspace reduction
  from \mc{\{\S,\F\},V_{-}} to \mc{\{\F\},V_{-}}.
\end{proof}

\section{Conclusion, and open problems: the ugly fragments}
\label{sec:concl}

We have almost completely separated the model-checking problem for Linear Temporal Logic
with respect to arbitrary combinations of temporal and propositional operators
into tractable and intractable cases.
We have shown that all tractable MC problems are at most \NL-complete or even easier to solve.
This exhibits a surprisingly large gap in complexity between tractable and intractable cases.
The only fragments that we have not been able to cover by our classification
are those where only the binary \XOR-operator is allowed.
However, it is not for the first time that this constellation has been difficult to handle,
see \cite{bhss06,bss+07}.
Therefore, these fragments can justifiably be called ugly.

The borderline between tractable and intractable fragments is somewhat diffuse
among all sets of temporal operators without \U.
On the one hand, this borderline is not determined by a single set of propositional operators
(which is the case for the satisfiability problem, see \cite{bss+07}).
On the other hand, the columns \cE\ and \cV\ do not, as one might expect, behave dually.
For instance, while \mc{\{\G\},\cV} is tractable, \mc{\{\F\},\cE} is not---although
\F\ and \G\ are dual, and so are \cV\ and \cE.

Further work should find a way to handle the open \XOR\ cases from this paper
as well as from \cite{bhss06,bss+07}.
In addition, the precise complexity of all hard fragments not in bold-face type in Table \ref{tab:overview}
could be determined.
Furthermore, we find it a promising perspective to use our approach for obtaining a fine-grained analysis
of the model-checking problem for more expressive logics, such as CTL, CTL*, and hybrid temporal logics.

  \bibliographystyle{alpha}

\begin{thebibliography}{BCRV03}

\bibitem[BCRV03]{bcrv03}
E.~B\"ohler, N.~Creignou, S.~Reith, and H.~Vollmer.
\newblock Playing with {B}oolean blocks, part {I}: Post's lattice with
  applications to complexity theory.
\newblock {\em SIGACT News}, 34(4):38--52, 2003.

\bibitem[BHSS06]{bhss06}
M.~Bauland, E.~Hemaspaandra, H.~Schnoor, and I.~Schnoor.
\newblock Generalized modal satisfiability.
\newblock In B.~Durand and W.~Thomas, editors, {\em STACS}, volume 3884 of {\em
  Lecture Notes in Computer Science}, pages 500--511. Springer, 2006.

\bibitem[BSS{\etalchar{+}}07]{bss+07}
M.~Bauland, T.~Schneider, H.~Schnoor, I.~Schnoor, and H.~Vollmer.
\newblock The complexity of generalized satisfiability for linear temporal
  logic.
\newblock In H.~Seidl, editor, {\em FoSSaCS}, volume 4423 of {\em Lecture Notes
  in Computer Science}, pages 48--62. Springer, 2007.

\bibitem[Dal00]{dal00}
V.~Dalmau.
\newblock {\em Computational Complexity of Problems over Generalized Formulas}.
\newblock PhD thesis, Department de Llenguatges i Sistemes Inform\`atica,
  Universitat Polit\'ecnica de Catalunya, 2000.

\bibitem[Lew79]{lew79}
H.~Lewis.
\newblock Satisfiability problems for propositional calculi.
\newblock {\em Mathematical Systems Theory}, 13:45--53, 1979.

\bibitem[Lyn77]{lynch77}
Nancy~A. Lynch.
\newblock Log space recognition and translation of parenthesis languages.
\newblock {\em Journal of the ACM}, 24(4):583--590, 1977.

\bibitem[Mar04]{mar04}
Nicolas Markey.
\newblock Past is for free: on the complexity of verifying linear temporal
  properties with past.
\newblock {\em Acta Inf.}, 40(6-7):431--458, 2004.

\bibitem[Nor05]{nor05}
G.~Nordh.
\newblock A trichotomy in the complexity of propositional circumscription.
\newblock In {\em LPAR}, volume 3452 of {\em Lecture Notes in Computer
  Science}, pages 257--269. Springer Verlag, 2005.

\bibitem[Pip97]{pip97b}
N.~Pippenger.
\newblock {\em Theories of Computability}.
\newblock Cambridge University Press, Cambridge, 1997.

\bibitem[Pnu77]{pnu77}
A.~Pnueli.
\newblock The temporal logic of programs.
\newblock In {\em FOCS}, pages 46--57. IEEE, 1977.

\bibitem[Pos41]{pos41}
E.~Post.
\newblock The two-valued iterative systems of mathematical logic.
\newblock {\em Annals of Mathematical Studies}, 5:1--122, 1941.

\bibitem[Rei01]{rei01}
S.~Reith.
\newblock {\em Generalized Satisfiability Problems}.
\newblock PhD thesis, Fachbereich Mathematik und Informatik, Universit\"at
  W\"urzburg, 2001.

\bibitem[RV03]{revo03}
S.~Reith and H.~Vollmer.
\newblock Optimal satisfiability for propositional calculi and constraint
  satisfaction problems.
\newblock {\em Information and Computation}, 186(1):1--19, 2003.

\bibitem[RW05]{rewa99-dt}
S.~Reith and K.~W. Wagner.
\newblock The complexity of problems defined by {B}oolean circuits.
\newblock In {\em MFI 99}. World Science Publishing, 2005.

\bibitem[Sav73]{sav73}
W.~J. Savitch.
\newblock Maze recognizing automata and nondeterministic tape complexity.
\newblock {\em Journal of Computer and Systems Sciences}, 7:389--403, 1973.

\bibitem[SC85]{sicl85}
A.~Sistla and E.~Clarke.
\newblock The complexity of propositional linear temporal logics.
\newblock {\em Journal of the ACM}, 32(3):733--749, 1985.

\bibitem[Sch07]{sch07}
H.~Schnoor.
\newblock {\em Algebraic Techniques for Satisfiability Problems}.
\newblock PhD thesis, University of Hannover, 2007.

\end{thebibliography}
\newcommand{\etalchar}[1]{$^{#1}$}


  \newpage
  \appendix
\section{Known Facts from Graph Theory}

\begin{lemma}
  \label{lem:existence_of_infinite_path_is_NL-hard}

  The following problem is \NL-hard.~
  Given a directed graph $G = (V,E)$ and a node $a \in V$, is there an infinite path
  in $G$ starting at $a$?
\end{lemma}

\begin{proof}
  We reduce from the graph accessibility problem (GAP), which is defined as follows.
  Given a directed graph $G = (V,E)$ and two nodes $a,b \in V$, is there a path
  in $G$ from $a$ to $b$? This problem is known to be \NL-complete~\cite{sav73}.

  For the reduction, consider an arbitrary instance $\mcinst{G,a,b}$ of GAP, where $G = (V,E)$
  and $a,b \in V$. Let $|V| = n$. We transform $G$ into a new graph $G'$ that consists
  of $n$ ``layers'' each of which contains a copy of the nodes from $V$. Whenever there
  is an edge from node $v$ to node $w$ in $G$, the new graph $G'$ will have edges from
  each copy of $v$ to the copy of $w$ on the next layer. This destroys all cycles
  from $G$. Now we add an edge from each copy of $b$ to the first copy of $a$.

  More formally, transform $\mcinst{G,a,b}$ into $\mcinst{G',a^1}$, where $G' = (V',E')$ with
  \begin{align*}
    V' & = \{v^i \mid v \in V \text{ and } 1 \le i \le n\},             \\
    E' & = \{(v^i,w^{i+1}) \mid (v,w) \in E \text{ and } 1 \le i < n\}
           ~\cup~\{(b^i,a^1) \mid 1 \le i \le n\}.
  \end{align*}

  It is easy to see that this transformation is a logspace reduction.
  Let the size of a graph be determined by the size of its adjacency matrix.
  Hence $G$ has size $n^2$, and $G'$ is of size $n^4$. Apart from the
  representation of $G'$, the only space required by the described transformation
  is spent for four counters that take values between $1$ and $n$. With their help,
  each bit of the new adjacency matrix is set according to the definition of $E'$,
  where only a look-up in the old adjacency matrix is required.

  It remains to prove the following claim.

  \begin{Claim}
    \label{claim:graph_problems}
    For each directed graph $G = (V,E)$ and each pair of nodes $a,b \in V$,
    there exists a path in $G$ from $a$ to $b$ if and only if there exists an
    infinite path in $G'$ starting at $a^1$.
  \end{Claim}

  \begin{Proofofclaim}[\ref{claim:graph_problems}]
    ``$\Rightarrow$''.~
    Suppose there is a path in $G$ from $a$ to $b$. 
    W.l.o.g.
    we can assume
    that no node occurs more than once on this path, $a$ and $b$ included. Hence there exist
    nodes $c_1,\dots,c_m \in V$ with $m \le n$ such that $c_1 = a$, $c_m = b$, and for each
    $i = 1,\dots,n-1$, $(c_i,c_{i+1}) \in E$. Due to its construction, $G'$ has the cycle
    $(c_1^1,~ c_2^2,~ \dots,~ c_{m}^m,~ a^1)$ that contains $a^1$. Hence $G'$ has an infinite path
    starting at $a^1$.

    \par\medskip\noindent
    ``$\Leftarrow$''.~
    Suppose there is an infinite path $p$ in $G'$ starting at $a^1$. Since $G'$ is finite, some node
    must occur infinitely often on $p$. This, together with the layer-wise construction of $G'$, implies
    that there are infinitely many nodes of layer 1 on $p$. Among layer-1 nodes, only $a^1$ has ingoing edges.
    Hence $a^1$ must occur infinitely often on $p$. Now the path from some occurrence of $a^1$ to the
    next is a cycle, where the predecessor node of $a^1$ must be some $b^m$. This implies that there
    is a path in $G'$ from $a^1$ to $b^m$. Due to the construction of $G'$, this corresponds to
    a path in $G$ from $a$ to $b$.
  \end{Proofofclaim}
\end{proof}

\end{document}